\documentclass[11pt,letterpaper]{article}
\usepackage[margin=1in]{geometry}
\usepackage{booktabs} 
\usepackage[ruled]{algorithm2e} 

\SetAlFnt{\small}
\SetAlCapFnt{\small}
\SetAlCapNameFnt{\small}
\SetAlCapHSkip{0pt}
\IncMargin{-\parindent}

\usepackage[numbers]{natbib}
\usepackage{array,multirow}
\usepackage[english]{babel}
\usepackage{stmaryrd}
\usepackage{algorithmic}
\usepackage{amssymb}
\usepackage{marginnote}
\usepackage{booktabs,tabularx,threeparttable}
\usepackage[normalem]{ulem}
\usepackage{hyperref}

\usepackage{tikz}
\usetikzlibrary{fit,backgrounds,positioning}
\usepackage{mathrsfs}
\usepackage{amsmath}
\usepackage{bbm}
\usepackage{graphicx}
\usepackage{stackengine}
\usepackage{amsthm}
\usepackage{subcaption}
\usepackage{mathtools}
\usepackage{thmtools}
\usepackage{thm-restate}
\usepackage[colorinlistoftodos]{todonotes}
\usepackage{amsfonts}
\usepackage[dvipsnames]{xcolor}
\usepackage{enumitem}
\usepackage{authblk}

\theoremstyle{acmplain}
\newtheorem{thm}{Theorem}[section]

\newtheorem{lm}[thm]{Lemma}
\newtheorem{obs}[thm]{Observation}
\newtheorem{cor}[thm]{Corollary}
\newtheorem{prop}[thm]{Proposition}
\newtheorem{cl}[thm]{Claim}
\theoremstyle{acmdefinition}

\newcommand{\PR}{\mathbb{P}}
\newcommand{\set}[1]{\left\{#1\right\}}

\newcommand{\floor}[1]{\left\lfloor#1\right\rfloor}

\newcommand{\scell}[2][c]{%
  \begin{tabular}[#1]{@{}c@{}}#2\end{tabular}}

\newcommand{\bb}[0]{\mathbb}

\usetikzlibrary{shapes.geometric}

\newcommand{\er}{er}

\newcommand{\blfootnote}[1]{%
  \begingroup
  \renewcommand{\thefootnote}{\fnsymbol{footnote}}
  \footnotetext[0]{\hspace*{-2.24em} #1}
  \endgroup
}

\title{Random Serial Dictatorship is $\sqrt{2}$-Envy-Free}

\author[1]{Frank~Connor}
\author[1]{Max~Dupré~la~Tour}
\author[1]{Louis-Roy~Langevin}
\author[2]{Vishnu~V.~Narayan}
\author[1]{Ndiamé~Ndiaye}
\author[1]{Neil~Rahman}
\author[1]{Adrian~Vetta}
\affil[1]{McGill~University}
\affil[2]{Indian~Institute~of~Technology~Bombay}

\date{July 3, 2026}

\begin{document}

\maketitle

\blfootnote{We are especially grateful to Simon Mauras, without whom this work would not have been possible. We thank Michal Feldman and Tomasz Ponitka for helpful discussions. This project was supported by NSERC Discovery Grant 2022-04191.

\noindent\texttt{\{frank.connor,louis-roy.langevin,vishnu.narayan,ndiame.ndiaye,neil.rahman\}@mail.mcgill.ca, maxduprelatour@gmail.com, adrian.vetta@mcgill.ca}.}
\vspace*{-1.5em}
\begin{abstract}
We analyze the house allocation problem, in which a set of agents must be matched to a set of objects for which the agents have cardinal utilities. A central mechanism for this problem is \emph{random serial dictatorship (RSD)}, which has long served as a canonical subject of study due to its simplicity and the existence of exact characterizations by its properties. Despite this extensive understanding, a basic quantitative question about the fairness of this mechanism remains unresolved. Although RSD is often viewed as fair ex ante, surprisingly, it is not envy-free in expectation. We quantify its deviation from envy-freeness via the \emph{envy-ratio}, which is defined as the maximum, over all instances and pairs of agents, of the ratio between an agent’s expected utility for another agent’s random object and for its own random object. Prior work shows a factor-$\sqrt{2}\approx{1.414}$ lower bound on the envy-ratio of the RSD mechanism. Our headline result is an upper bound that matches this constant factor, proving that RSD is $\sqrt{2}$-envy-free in the house allocation problem.

We further analyze the two natural extensions of RSD (the so-called \emph{randomized round-robin} mechanism and the \emph{iterated-RSD} mechanism) to settings with unequal numbers of agents and objects and with more general valuation classes. For additive valuations, we establish constant-factor bounds on the envy-ratio of both mechanisms, showing that this ratio increases to at least $1.5$ and at most $1+(1/\sqrt{2})\approx 1.707$ for randomized round-robin, but remains exactly $\sqrt{2}$ for iterated-RSD. For submodular valuations, we prove constant-factor upper and lower bounds for both mechanisms, leaving only a small constant gap in both cases. For the more general classes of XOS and subadditive valuations, we present a tight analysis of the envy-ratio of both mechanisms, showing that this ratio is unbounded in the number of agents. These results provide the first tight or nearly tight quantitative guarantees on the extent to which random serial dictatorship and its natural generalizations approximate envy-freeness.
\end{abstract}

\section{Introduction}

How should a collection of 
objects be fairly matched to a set of 
agents, who have varying preferences over the objects, so that each agent is assigned exactly one object? This 
matching problem, typically called the \textit{house allocation problem}, arises in a multitude of situations such as the assignment of dorms to students, of hotel rooms to guests, and of parking spaces to employees.

Perhaps the simplest solution to the house allocation problem is the following: the agents line up in some order, and each agent in turn selects its most valuable object among the remaining objects until every agent is allocated an object. This straightforward mechanism, called a \textit{serial dictatorship}, has been extensively studied over the past few decades. Despite its simplicity, this mechanism has several desirable properties, such as \textit{strategyproofness} (no agent can benefit by misreporting its preferences), \textit{non-bossiness} (no agent can change another agent's object without changing its own object), and \textit{neutrality} (when the objects are permuted, the assignment is permuted accordingly). In fact, serial dictatorships are the only deterministic mechanisms that simultaneously satisfy all of these properties (\citet*{svensson1999strategy}). Similar axiomatic characterizations of serial dictatorships by other desirable properties have been obtained in the economics literature~\cite{papai2001strategyproof,ehlers2003coalitional}.

Conspicuously absent from these characterizations are mentions of fairness, since serial dictatorships are patently unfair. \textit{Envy-freeness}, the most prominent fairness property, requires that no agent has a strictly larger value for the object assigned to another agent than for its own object. It is easy to see that serial dictatorships are not envy-free; for example, in instances with one universally desirable object, the agent that chooses first is envied by every other agent.

What, then, can we do to add fairness to this mechanism? The most natural approach to circumvent the inherent unfairness is to uniformly randomize the order of the agents before they select from the objects. The resulting mechanism, dubbed \textit{random serial dictatorship} (RSD) or \textit{random priority}, has also been extensively studied, resulting in a selection of highly influential works in theoretical economics such as~\citet*{abdulkadirouglu1998random} and~\citet*{bogomolnaia2001new}. Additionally, RSD and its variations have seen widespread adoption in the real world. For example, they are among the most commonly used mechanisms for the assignment of dormitory rooms to students~\cite{abdulkadirouglu2003school,sonmez2011matching}, have been routinely employed for allocating overdemanded tickets for shows and games~\cite{moulin2019fair}, and proposed as a possible solution for kidney exchanges, for which they achieve strong fairness and efficiency outcomes on real-world data~\cite{fang2015randomized}. Owing to this central role, obtaining a precise characterization of the properties of the RSD mechanism is of fundamental importance.

Many desirable properties of serial dictatorships continue to hold for their randomized extension. In particular, RSD is strategyproof~\cite{bogomolnaia2001new}, which implies that an agent can never benefit in any outcome, and therefore benefit in expectation, by reporting its preferences untruthfully. Moreover, RSD is clearly ex-ante ``fair'', since every agent has the same probability of appearing in every given position, and every agent selects its favorite available object in every possible permutation. Indeed, it is natural to conclude that no agent envies any other agent in expectation, since an agent can never improve its outcome by reporting another agent's preferences.

But is RSD actually envy-free in expectation? The strong intuitive appeal of the above argument makes this claim surprisingly ubiquitous: this conclusion was not only reached quickly in our informal conversations with other researchers, and not just reproduced by large-scale language models trained on the literature,\footnote{\emph{``Yes, random serial dictatorship is envy-free in expectation''}, ChatGPT, early 2025.} but until recently even reiterated by the Wikipedia article about this mechanism.\footnote{\href{https://en.wikipedia.org/w/index.php?title=Random_priority_item_allocation&oldid=1136792708}{``Random priority item allocation'', Wikipedia, February 1, 2023}.} Bizarrely, however, it is actually possible for an agent to prefer, in expectation, the random object of another agent to its own random object! This unusual property was first observed by~\citet*{bogomolnaia2001new}, who presented an instance (see Section~\ref{sec:prelims}) with three agents and three objects that exhibits this behavior.\footnote{While \cite{bogomolnaia2001new} and other similar works on the house allocation problem often consider ordinal preferences, we are interested in the cardinal setting.}

Despite its canonical position as the central mechanism for the house allocation problem, a basic question 
that follows from the above observation surprisingly remained unanswered:
\begin{center}
    \em In expectation, precisely what approximation of envy-freeness does RSD guarantee?
\end{center}
In other words, how large can the above envy become? More formally, we want to determine the \textit{envy-ratio} of the RSD mechanism, defined as the largest ratio between an agent’s expected value for another agent’s random object and its expected value for its own random object (see Section~\ref{sec:prelims} for a formal definition). Prior works on this topic present lower bounds on the envy-ratio of RSD. An example in \citet*{freeman2020best} establishes a $\frac{20}{19}$-factor lower bound on this ratio. \citet*{feldman2024breaking} remark that a modification of the example of \citet*{bogomolnaia2001new} gives a lower bound of~$\frac{5}{4}$, and improve this bound with a sequence of examples showing the envy-ratio is at least~$\sqrt{2}$.



In this work, our headline result is a matching upper bound, uncovering that the envy-ratio of the RSD mechanism is exactly $\sqrt{2}$, and implying that RSD does, in fact, guarantee a factor-$\sqrt{2}$ approximation to envy-freeness in the house allocation problem.

Beyond this classical unit-demand setting, in many applications, such as the assignment of players to teams in sports drafts, the number of objects exceeds the number of agents and the agents themselves have more general valuation functions. Consequently, we also extend our analysis of RSD to these more general settings. 
To do so, we remark that there are two natural extensions of RSD where the agents select the objects in rounds until no objects remain.  
In the first method, called {\em randomized round-robin}~\cite{freeman2020best}, 
the agents select objects via the same random order in each round. In the second method, called {\em iterated-RSD},
the agent order is rerandomized in each round.
For additive valuation functions, we 
show that the envy-ratio increases 
for randomized round-robin but remains exactly $\sqrt{2}$ for iterated-RSD.

We then study non-additive valuation functions.
For submodular valuations, we present constant-factor upper and lower bounds on the envy-ratio of these mechanisms.
For XOS and subadditive valuations, we present a tight analysis showing that the envy-ratio is unbounded. Our results are formally presented in the following subsection.

\subsection{Our Results}
In Section~\ref{sec:prelims} we present the model and motivating examples. In Section~\ref{sec:rsd} we prove our headline result that that RSD is exactly $\sqrt{2}$-envy-free for unit-demand agents.
\begin{restatable}{thm}{thmmain}\label{thm:main}
    The envy-ratio of random serial dictatorship is $\sqrt{2}$ in the house allocation problem.
\end{restatable}
While the number $|M|$ of objects equals the number $|N|$ of agents in the house allocation problem, this guarantee holds even when $|M|$ differs from $|N|$ as long as the agents are unit-demand.\footnote{To see this, when $|M|<|N|$, it suffices to add $|N|-|M|$ dummy objects, each of which has a value of zero for every agent, reducing to the $|M|=|N|$ case. When $|M|>|N|$, our analysis in the proof of Theorem~\ref{thm:main} continues to apply when the agents are unit-demand.} We then study the more general fair division problem where $|M|$ is greater than $|N|$, and where the agents have more general valuation functions. The necessitates the examination of two multiple-round extensions of RSD. In the {\em randomized round-robin} mechanism~\cite{freeman2020best}, a single uniformly random order is fixed in advance for all rounds. In the {\em iterated-RSD} mechanism  
the agent order is independently rerandomized at the start of each round. 
In Section~\ref{sec:additive}, we study these extensions in the case of additive valuation functions.
For the iterated-RSD mechanism the envy-ratio of the resulting mechanism remains equal to $\sqrt{2}$.
\begin{restatable}{thm}{thmirsdadd}\label{thm:irsdadd}
For additive valuations, the envy-ratio of iterated-RSD is $\sqrt{2}\approx1.414$.
\end{restatable}
Interestingly, in contrast the envy-ratio deteriorates
for the randomized round-robin mechanism.
Specifically, the envy-ratio increases to at least $\frac32$, but is still bounded above by a constant factor of $1+(1/\sqrt{2})\approx 1.707$.
\begin{restatable}{thm}{thmrrradd}\label{thm:rrradd}
For additive valuations, the envy-ratio of randomized round-robin is at least $3/2 =1.5$ and at most $1+(1/\sqrt{2})\approx 1.707$.
\end{restatable}


In Section~\ref{sec:nonadditive}, we extend our analysis beyond additive valuation functions. 
To do this, for non-additive valuations, we make the assumption that each time an agent selects an object it
chooses an available object with the highest marginal value.

For the case of submodular valuation functions we still obtain strong guarantees: the envy-ratio of both mechanisms is still a small constant.
\begin{restatable}{thm}{thmirsdsub}\label{thm:irsdsub}
For submodular valuations, the envy-ratio of iterated-RSD is at least $\frac47(1+2\sqrt{2})\approx2.188$ and at most $1+\sqrt{2}\approx 2.414$.
\end{restatable}
\begin{restatable}{thm}{thmrrrsub}\label{thm:rrrsub}
For submodular valuations, the envy-ratio of randomized round-robin is at least $9/4=2.25$ and at most $2+1/\sqrt{2}\approx 2.707$.
\end{restatable}

Unfortunately, constant bounds on the envy-ratio are no longer achieved when we move beyond submodular functions to the broader class
of subadditive valuation functions. Indeed, for XOS and subadditive valuations, we prove a general theorem showing that the envy-ratio of each mechanism is unbounded in $|N|$, and in fact is instead exactly equal to the number of rounds $r$ in the execution of the mechanism.
\begin{restatable}{thm}{thmgeneral}\label{thm:general}
For XOS and subadditive valuations, the envy-ratio of both randomized round-robin and iterated-RSD is exactly the number of rounds $r$ in the execution of the mechanism.
\end{restatable}




In Section~\ref{sec:discussion}, we discuss some implications of our work and open problems.



\subsection{Related Work}
As mentioned in the introduction, the two seminal works that consider the fairness of random serial dictatorship are \citet*{abdulkadirouglu1998random} and \citet*{bogomolnaia2001new}. The latter of these works famously shows that RSD satisfies a set of desirable properties---namely, equal treatment of equals, ex-post (ordinal) efficiency, and strategyproofness---and in fact, for instances with three agents and three objects is characterized by this set of properties. In this direction, a line of research in theoretical economics finds (e.g.~\cite{bade2020random,basteck2025axiomatization}), or in some cases, disproves (e.g.~\cite{basteck2025constrained}) similar axiomatic characterizations of RSD by various sets of desirable properties.

In the more general setting where agents can have cardinal utilities, there is a large body of literature that seeks to design deterministic or randomized mechanisms that \emph{approximately} guarantee various standard fairness criteria depending on structural assumptions about the agent's valuations. These criteria include envy-freeness and its relaxations such as EF1 and EFX (see e.g.~\cite{amanatidis2020multiple,feldman2024breaking}), and share-based fairness notions such as the maximin share, the any-price share or the quantile share (e.g.~\cite{budish2011combinatorial,kurokawa2018fair,babaioff2021fair,babichenko2024fair}). A parallel line of research studies the factor by which RSD approximates the optimal \emph{social welfare} in both the ordinal (e.g.~\cite{bhalgat2011social}) and cardinal (e.g.~\cite{filos2014social}) settings.



\section{Preliminaries} \label{sec:prelims}

In the house allocation problem, there is a set $N$ of agents and a collection $M$ of objects/items, where $|M|=|N|$. Each agent $A$ has a non-negative value $v_A(j)$ for each item $j$. We assume that each agent $A$ has a total order over $M$ that is coherent with its valuation function, so that if $v_A(j) > v_A(j')$ then $j$ is ranked above $j'$, and ties are broken according to a fixed rule. This ensures that each agent has a unique most-preferred item among any subset of remaining items.

\paragraph{Serial Dictatorship.} Let $\Omega$ be the set of all $|N|!$ permutations of the agents in $N$. For any fixed permutation $\sigma \in \Omega$, the serial dictatorship with permutation $\sigma$ operates as follows: for $k = 1, \dots, |N|$, assign to agent $\sigma(k)$ its most-preferred item among those that remain. This procedure is deterministic and well defined, and terminates when all items have been assigned.

\paragraph{Random Serial Dictatorship.} The random serial dictatorship (RSD) mechanism is the distribution over allocations obtained by drawing a permutation $\hat\sigma$ uniformly at random from $\Omega$ and applying the above serial dictatorship procedure.\\

As stated in the introduction, RSD is not envy-free in expectation for the house allocation problem. To show this, we require the following definitions.
A deterministic allocation $X$ is \emph{envy-free} if no agent envies another agent's bundle of items in $X$, that is, $v_A(X_A) \geq v_A(X_B)$, for every ordered pair $(A,B)$ of agents. 
A random allocation $X$ is \emph{envy-free in expectation} if $\mathbb{E}[v_A(X_A)] \geq \mathbb{E}[v_A(X_B)]$, for every ordered pair $(A,B)$ of agents. We remark that for the house allocation problem, any bundle $X_A$
will consist of one item. However, for the agents with more general valuations encountered in Sections~\ref{sec:additive} and~\ref{sec:nonadditive} such bundles may contain many items.

\paragraph{Envy-Ratio.} Consider a fair division instance and a corresponding (possibly randomized) allocation $X$, such that for every agent $A$, $\mathbb{E}[v_A(X_A)]>0$.
For any pair of distinct agents $A,B$, the envy-ratio of $A$ toward $B$ is $\mathbb{E}[v_A(X_B)]/\mathbb{E}[v_A(X_A)]$, where the expectation is taken over the randomness of the allocation $X$. The envy-ratio of an allocation $X$ is the maximum of this quantity over all ordered pairs $(A,B)$ of agents. The envy-ratio of a (possibly randomized) mechanism is the supremum, over all input instances on which the mechanism is defined, of the envy-ratio of the allocation produced by the mechanism on that instance. For RSD, we restrict our attention to instances in which each agent assigns a positive value to at least one item, ensuring that every agent has a positive expected value and that the envy-ratio is well defined. \\

Demonstrating that RSD is not envy-free requires an instance with just three agents (A, B and C) and three items ($a$, $b$ and $c$). Similar instances are presented in~\citet*{bogomolnaia2001new} and \citet*{freeman2020best}. \citet*{feldman2024breaking} remark that a modification of the instances in these works gives a lower bound of $5/4$ on the envy-ratio of the RSD mechanism. This example is shown in Table~\ref{tab:lb-54}:
the value each agent has for each item is shown on the LHS; the item each agent is allocated by each of the six orderings (permutations) used by RSD is shown on the RHS.
\begin{table}[!h]
    \begin{minipage}{.4\linewidth}
      \centering
        \begin{tabularx}{0.6\textwidth}{c|ccc}
           & \multicolumn{3}{c}{{\bf Item}} \\
           \cline{2-4}
        {\bf Agent} & $a$ & $b$ & $c$ \\
        \cline{1-4}
        {\tt A} & 1 & $1-\varepsilon$ & 0 \\
        {\tt B} & $1-\varepsilon$ & 1 & 0 \\
        {\tt C} & 1 & 0 & $1-\varepsilon$
        \end{tabularx}
    \end{minipage}
    \quad
    \begin{minipage}{.4\linewidth}
      \centering
         \begin{tabularx}{0.75\textwidth}{c|XXX}
           & \multicolumn{3}{c}{{\bf Agent}} \\
           \cline{2-4}
        {\bf Ordering} & {\tt A} & {\tt B} & {\tt C} \\
        \cline{1-4}
        {\tt ABC} & $a$ & $b$ & $c$ \\
        {\tt ACB} & $a$ & $b$ & $c$ \\
        {\tt BAC} & $a$ & $b$ & $c$ \\
        {\tt BCA} & $c$ & $b$ & $a$ \\
        {\tt CAB} & $b$ & $c$ & $a$ \\
        {\tt CBA} & $c$ & $b$ & $a$
    \end{tabularx}
    \end{minipage} 
       \caption{Valuations for RSD instance demonstrating a factor-$\frac{5}{4}$ lower bound on the envy-ratio}\label{tab:lb-54}
\end{table}

We are interested in agent~$A$'s envy toward agent~$B$. 
Observe that the RSD mechanism allocates to agent~$A$ the item~$a$ in three of the orderings and the item~$b$ in one ordering. Thus the expected value $A$ has for its 
own random item is $E(v(X_A))=\frac16 (3+(1-\varepsilon))=\frac16 (4-\varepsilon)$. On the other hand, the RSD mechanism allocates to agent~$B$ the item 
$b$ in five of the orderings. Consequently, the expected value $A$ has for $B$'s random item is $E(v(X_B))=\frac16 (5(1-\varepsilon))=\frac16 (5-5\varepsilon)$. 
Taking $\varepsilon\rightarrow0$, we see that agent~$A$ is envious of
agent~$B$. Indeed, it gives that envy-ratio of RSD is at least $5/4$.

\begin{cor}{{\em \cite{bogomolnaia2001new, feldman2024breaking}}}
The envy-ratio of RSD is at least $5/4$ in the house allocation problem.
\end{cor}

What about upper bounds on the envy-ratio of RSD? The following simple unpublished argument shows that the envy-ratio is at most $2$ for the house allocation problem.
\begin{prop}\label{prop:simpleupper}
The envy-ratio of RSD is at most $2$ in the house allocation problem.
\end{prop}
\begin{proof}
Fix two distinct agents $A$ and $B$. Partition the set $\Omega$ of permutations of $N$ as
$\Omega = \Omega_A \sqcup \Omega_B$, where $\Omega_A$ (resp.\ $\Omega_B$) consists of permutations
in which $A$ appears before $B$ (resp.\ $B$ appears before $A$).
There is a bijection between $\Omega_A$ and $\Omega_B$ given by swapping the positions of $A$ and $B$;
for $\sigma \in \Omega$.

Fix $\sigma_A \in \Omega_A$ and let $\sigma_B \in \Omega_B$ be the permutation given by swapping $A$ and $B$ in $\sigma_A$.
Under $\sigma_A$, agent $A$ selects before $B$, so $A$'s value for its own item is at least its value
for the item assigned to $B$ under $\sigma_A$.
Under $\sigma_B$, the relative order of all agents other than $A$ and $B$ is unchanged, so the agents
preceding $B$ in $\sigma_B$ are exactly those preceding $A$ in $\sigma_A$ and they choose the same items.
Thus, the set of items available to $A$ under $\sigma_A$ is identical to the set available to $B$
under $\sigma_B$, implying that $A$'s value for its item under $\sigma_A$ is at least its value for
$B$'s item under $\sigma_B$.

Pairing permutations via the above bijection and using the uniform randomness of RSD, we obtain that
$A$'s expected value for $B$'s allocated item is at most twice its expected value for its own item.
Hence, the envy-ratio of Random Serial Dictatorship is at most $2$.
\end{proof}

\subsection{Preliminary Lemmas}

We will use the following stability fact for serial dictatorships when an agent is inserted into the prefix of a permutation.

\begin{lm}\label{lm:airplane}
Let $\sigma$ be a permutation in $\Omega$. Fix $r<|N|$, and let $z_i$ be the item chosen by $\sigma(i)$ under serial dictatorship with order $\sigma$, for $i=1,\dots,k$.
Let $A$ be an agent that appears after position $k$ in $\sigma$, and let $\sigma'$ be obtained by removing $A$ from $\sigma$ and inserting it into some position $p\in\{1,\dots,k\}$, preserving the relative order of all other agents.
Then the set of items chosen by the first $k+1$ agents under $\sigma'$ contains $\{z_1,\dots,z_k\}$.

\smallskip
\end{lm}

\begin{proof}

We show by induction on $i=1,\dots,k$ that $z_i$ is selected within the first $i+1$ steps under $\sigma'$. If $i<p$, then the first $i$ agents are the same in $\sigma$ and $\sigma'$, so $\sigma(i)$ faces the same available set and selects $z_i$ at step $i$ under $\sigma'$. Now fix $i\ge p$ and assume $z_1,\dots,z_{i-1}$ have all been removed by step $i$ under $\sigma'$. Under $\sigma$, agent $\sigma(i)$ selects $z_i$ at step $i$. Under $\sigma'$, agent $\sigma(i)$ acts at step $i+1$, after $i$ selections, and by the induction hypothesis all of $\{z_1,\dots,z_{i-1}\}$ have already been removed. Hence if $z_i$ is still available at step $i+1$, then $\sigma(i)$ still selects $z_i$ (which remains its most-preferred item). Otherwise, $z_i$ has already been removed. This completes the induction.
\end{proof}

We will also need the following simple observation.
\begin{obs}\label{obs:fraction}
For any positive real numbers $a_1,\ldots,a_n$ and non negative real numbers $b_1,\ldots,b_n$,
\begin{equation*}
    \frac{\sum_{i\in[n]} b_i}{\sum_{i\in[n]} a_i} \le \max_{i\in[n]} \frac{b_i}{a_i}.
\end{equation*}
\end{obs}

\section{Random Serial Dictatorship}\label{sec:rsd}

In this section we analyze the RSD mechanism and exactly determine its envy-ratio for the house allocation problem. \citet*{feldman2024breaking} present a sequence of instances showing that this ratio is at least $\sqrt{2}$. Here, we present an intricate analysis that upper-bounds the envy-ratio of RSD at a matching factor of $\sqrt{2}$, giving us our headline result.

\thmmain*

\subsection{Technical Overview}

Since our goal is to bound the envy-ratio of the mechanism for an arbitrary instance, it suffices to focus on two agents, which we label $A$ and $B$, and bound the envy-ratio of $A$ toward $B$. This ratio depends on $A$'s cardinal valuation and on the ordinal preferences of $B$ and the remaining agents (which determine their picking order over the items). Let $\mathcal{C}$ denote the set of remaining agents.

At a high level, the proof proceeds by transforming the analysis of the envy-ratio into an optimization problem over a continuous domain. We first define two vectors $\mathbf{x}$ and $\mathbf{y}$ that are indexed by the possible positions of $A$ and $B$ within the ordering of the remaining agents $\mathcal{C}$, and that encode the availability of items for $A$ and $B$ at any of these positions. We then define a cone $K$ in the space of vector pairs $(\mathbf{x},\mathbf{y})$ that captures the structural constraints shared by all instances of the house allocation problem, i.e., every instance corresponds to a point $(\mathbf{x},\mathbf{y}) \in K$.

This transformation serves two important purposes. First, it allows us to show that, to bound the envy-ratio, it suffices to fix some worst-case ordering of the agents in $\mathcal{C}$ and analyze only permutations over all agents that are consistent with this ordering of $\mathcal{C}$. Second, and crucially, it allows us to replace the analysis over instances of the house allocation problem with an analysis over the space of vector pairs $(\mathbf{x},\mathbf{y})$. We define a function $\mathcal{U}(\mathbf{x},\mathbf{y})$ on the domain $K$, and show that for every instance of the house allocation problem, the envy-ratio of $A$ toward $B$ is upper bounded by $\mathcal{U}(\mathbf{x},\mathbf{y})$ evaluated at the corresponding $(\mathbf{x},\mathbf{y})$ for that instance.

We can now analyze this optimization problem directly. Although $K$ is a relaxation that includes infeasible points not corresponding to any instance, we show using a convexity argument that any \emph{maximizer} $(\mathbf{x}^*,\mathbf{y}^*)$ of $\mathcal{U}$ over $K$ corresponds to a structured instance of the house allocation problem in which $A$ values every item at either $0$ or $1$. This transformation thus captures the worst-case behavior of the envy-ratio using these structured binary instances, which are easier to analyze. Finally, we show that for these structured instances the envy-ratio is at most $\sqrt{2}$. For completeness, we also present a simplified version of the sequence of instances from \cite{feldman2024breaking} giving us the matching lower bound of $\sqrt{2}$.

We remark that the example in Table~\ref{tab:lb-54} giving the factor-$\frac{5}{4}$ lower bound highlights an essential aspect of our proof. If we focus on the permutations $ACB$ and $BCA$, in which the positions of agents $A$ and $B$ are swapped, we may observe that $B$ receives an item of high value in both permutations, but that $A$ fails to obtain a high-valued item in the permutation in which it appears after $B$. This happens because $B$ \textit{does not select} the same item as $A$ when it goes first, but instead selects an alternative high-valued item that is not demanded by agent~$C$, leading to the unavailability of any high-valued items for $A$ when it appears later in the ordering. The specific quantity we wish to analyze is the maximum extent to which this kind of effect can increase the envy-ratio of RSD.

\subsection{A $\sqrt{2}$ Upper Bound on the Envy-Ratio}\label{sec:upper-single}

We focus on two distinct agents, which we label $A$ and $B$, and bound the envy-ratio of $A$ toward $B$. Let $\mathcal{C}$ denote the set of remaining agents, and let $n = |\mathcal{C}|$. Let $v$ be agent~$A$'s valuation function.
For any permutation $\sigma \in \Omega$, let $v_A(\sigma)$ denote the value that agent $A$ obtains when the serial dictatorship mechanism is run according to permutation $\sigma$, and let $v_B(\sigma)$ denote the value that $A$ has for $B$'s allocation when the serial dictatorship mechanism is run according to permutation $\sigma$. 
Let $\hat\sigma$ be the random permutation selected by the RSD mechanism. Then $v_A(\hat\sigma)$ and $v_B(\hat\sigma)$ are random variables representing, respectively, agent $A$’s value for its own allocated item and its value for the item allocated to agent $B$. Because we no longer consider $B$'s valuation again, we emphasize that from this point forth, $v_A$ and $v_B$ no longer refer to $A$'s and $B$'s valuations, but instead to $A$'s value for its own item and for $B$'s item respectively. 

We are interested in bounding the envy-ratio between agents $A$ and $B$, defined as
\[\er(\Omega) = \frac{\mathbb{E}[v_B(\hat\sigma)]}{\mathbb{E}[v_A(\hat\sigma)]},\]
where the expectation is over the randomness of the RSD mechanism. More generally, for any subset $\Omega^* \subseteq \Omega$ such that $\mathbb{E}[v_A(\hat\sigma)|\hat\sigma\in \Omega^*]>0$, we define the conditional envy-ratio
\[\er(\Omega^*)=\frac{\mathbb{E}[v_B(\hat\sigma)|\hat\sigma\in \Omega^*]}{\mathbb{E}[v_A(\hat\sigma)|\hat\sigma\in \Omega^*]}.\]

We now present the main result of this section, which is that the envy-ratio of RSD is bounded above by $\sqrt{2}$ in the house allocation problem. Specifically, we restate the upper bound portion of Theorem~\ref{thm:main} using our notation.

\begin{thm}\label{thm:upper}
    $\er(\Omega)\leq \sqrt2$.
\end{thm}

We denote by $\Omega'$ the set of permutations of $\mathcal{C}$. Given $\sigma' \in \Omega'$, let $\Omega(\sigma')$ denote the set of permutations $\sigma \in \Omega$ such that the order of the agents in $\mathcal{C}$ induced by $\sigma$ is the same as their order in $\sigma'$; in this case, we say that $\sigma$ is consistent with $\sigma'$. Note that this implies $\Omega = \bigcup_{\sigma' \in \Omega'} \Omega(\sigma')$.

The following lemma allows us to fix a permutation $\sigma' \in \Omega'$ of the agents in $\mathcal{C}$ and bound only the conditional envy-ratio over permutations consistent with $\sigma'$.
\begin{restatable}{lm}{lmorder}\label{lm:order}
    $\er(\Omega)\leq \max_{\sigma'\in \Omega'} \er(\Omega(\sigma'))$.    
\end{restatable}
\begin{proof}
Recall that $\hat\sigma$ is the random permutation selected by the RSD mechanism. We have
    \begin{align*}
    \er(\Omega)
     &= \frac{\sum_{\sigma\in \Omega}\PR(\hat\sigma = \sigma)
      \cdot v_B(\sigma)}{\sum_{\sigma\in \Omega}\PR(\hat\sigma = \sigma)
      \cdot v_A(\sigma)} = \frac{\sum_{\sigma'\in \Omega'}\PR(\hat\sigma\in \Omega(\sigma'))
      \cdot \mathbb{E}[v_B(\hat\sigma)|\hat\sigma\in \Omega(\sigma')]}{\sum_{\sigma'\in \Omega'}\PR(\hat\sigma\in \Omega(\sigma'))
      \cdot \mathbb{E}[v_A(\hat\sigma)|\hat\sigma\in \Omega(\sigma')]}.
    \end{align*}
Since $\PR(\hat\sigma \in \Omega(\sigma'))$ does not vary with the choice of $\sigma'$, we have
 \begin{align*}
    \er(\Omega)
          = \frac{\sum_{\sigma'\in \Omega'}\mathbb{E}[v_B(\hat\sigma)|\hat\sigma\in \Omega(\sigma')]}{\sum_{\sigma'\in \Omega'}\mathbb{E}[v_A(\hat\sigma)|\hat\sigma\in \Omega(\sigma')]}
            \le \max_{\sigma'\in \Omega'} \frac{\mathbb{E}[v_B(\hat\sigma)|\hat\sigma\in \Omega(\sigma')]}{\mathbb{E}[v_A(\hat\sigma)|\hat\sigma\in \Omega(\sigma')]}
           = \max_{\sigma'\in \Omega'} \er(\Omega(\sigma')),
    \end{align*}
where the inequality follows from Observation~\ref{obs:fraction}. We implicitly used the fact that for all $\sigma' \in \Omega'$, there exists a permutation $\sigma \in \Omega(\sigma')$ in which $A$ picks first and therefore receives an item with strictly positive value. Hence $\mathbb{E}[v_A(\hat\sigma)\mid \hat\sigma \in \Omega(\sigma')]>0$, and all the ratios are well defined.
\end{proof}

By Lemma~\ref{lm:order}, it suffices to bound $\max_{\sigma' \in \Omega'} \er(\Omega(\sigma'))$. Fix a permutation $\sigma' = (C_1, C_2, \dots, C_n) \in \Omega'$ of the agents in $\mathcal{C}$. In the remainder of the proof, we establish the following lemma, which, when combined with Lemma~\ref{lm:order}, directly yields Theorem~\ref{thm:upper}.

\begin{lm}\label{lm:c'upper_bound}
    $\er(\Omega(\sigma')) \leq \sqrt2$.
\end{lm}

For any permutation $\sigma \in \Omega$ consistent with $\sigma'$, we define the index of agent $A$ to be the number of agents from $C$ that precede $A$ in $\sigma$. If exactly $\ell$ agents from $C$ precede $A$, we write $ind(A)=\ell$. For indices $i,j$ with $i \le j$, we write $ind(\sigma)=(i,j)$ if $\{ind(A),ind(B)\}=\{i,j\}$. For each pair $(i,j)$ with $0 \le i \le j \le n$, there are exactly two permutations $\sigma_A(i,j)$ and $\sigma_B(i,j)$ consistent with $\sigma'$ such that $ind(\sigma_A(i,j)) = ind(\sigma_B(i,j)) = (i,j)$. We define $\sigma_A(i,j)$ to be the permutation in which $A$ precedes $B$, and $\sigma_B(i,j)$ to be the permutation in which $B$ precedes $A$. Then $\Omega(\sigma') = \bigcup_{0 \leq i \leq j \leq n} \{\sigma_A(i,j),\sigma_B(i,j)\}$.

Now, suppose that the agents in $\mathcal{C}$ are allowed to select items in a serial dictatorship without agents $A$ and $B$ and ordered by $\sigma' = (C_1,C_2,\dots,C_n)$. We denote by $z_{\ell}$ the item selected by agent $C_{\ell}$ in this setting. Let $\mathbf{x}=(x_0,\ldots,x_n)$ and $\mathbf{y}=(y_0,\ldots,y_n)$ be two vectors defined as follows. For each $\ell\in[n]$, $x_\ell$ and $y_\ell$ are respectively $A$'s favorite item and second-favorite item in
$M\setminus \{z_1, \dots, z_{\ell}\}$.

The following lemma summarizes the outcome of agents $A$ and $B$ under the two permutations $\sigma_A(i,j)$ and $\sigma_B(i,j)$. To simplify the notation, we will say that an agent "gets $x$ in $\sigma$" if they select the item $x$ when we run the serial dictatorship with permutation $\sigma$.

\begin{lm}\label{lm:z}
Let $(i,j)$ be any pair with $0 \le i \le j \le n$. Agent $A$ always selects $x_i$ in $\sigma_A(i,j)$, and either $x_j$ or $y_j$ in $\sigma_B(i,j)$. There are three cases:

\begin{enumerate}
    \item $B$ gets $x_i$ in $\sigma_B(i,j)$, $A$ gets $x_j$ in $\sigma_B(i,j)$. Then, the value of $B$'s item in $\sigma_A(i,j)$ is at most $v(x_j)$.
    \item $B$ gets $x_i$ in $\sigma_B(i,j)$, $A$ gets $y_j$ in $\sigma_B(i,j)$. Then, the value
    of $B$'s item in $\sigma_A(i,j)$ is at most $v(y_j)$.
    \item $B$ does not get $x_i$ in $\sigma_B(i,j)$. Then, the value of $B$'s item in $\sigma_A(i,j)$ is at most $v(x_j)$.
\end{enumerate}
\end{lm}

\begin{proof}
Consider first the permutation $\sigma_A(i,j)$. Agent $A$ picks their favorite item after the first $i$ agents of $\mathcal{C}$ have picked. By definition, $A$'s choice is $x_i$. Agent $B$ then picks when $j+1$ items have been removed (namely, the $j$ items chosen by the agents of $\mathcal{C}$ and the item chosen by $A$). By (i) in Lemma~\ref{lm:airplane}, all items in $\{z_1,\dots,z_j\}$ have already been removed at this point. Hence the value (according to $A$) of the item chosen by $B$ is at most $v(x_j)$.

We now turn to the permutation $\sigma_B(i,j)$. Agent $B$ picks after the first $i$ agents of $\mathcal{C}$, so either $B$ gets $x_i$, or the value of $B$'s item (according to $A$) is at most $v(y_i)$. Agent $A$ picks when $j+1$ items have been removed. Using Lemma~\ref{lm:airplane}, all items in $\{z_1,\dots,z_j\}$ have been removed, along with one additional item. Therefore, after removing $\{z_1,\dots,z_j\}$, the two most-preferred remaining items of $A$ are $x_j$ and $y_j$, so $A$ receives either $x_j$ or $y_j$.

This proves the first and third case. It remains to bound the value of $B$'s item in $\sigma_A(i,j)$ in the second case. If $B$ gets $x_i$ in $\sigma_B(i,j)$, then running the serial dictatorship with $\sigma_A(i,j)$ and with $\sigma_B(i,j)$ yields identical first $j+1$ items, because $A$ and $B$ select the same item when they appear for the first time after the first $i$ picks. Therefore, in the second case, if $A$ cannot get $x_j$ in $\sigma_B(i,j)$ and instead gets $y_j$, then the value of $B$’s item in $\sigma_A(i,j)$ is at most $v(y_j)$.
\end{proof}

We now partition the pairs $(i,j)$ into three types, corresponding to the three cases of Lemma~\ref{lm:z}. We define:
\begin{enumerate}
\item $T_{x,x}$, the set of all pairs $(i,j)$ such that $B$ gets $x_i$ in $\sigma_B(i,j)$ and $A$ gets $x_j$ in $\sigma_B(i,j)$.
\item $T_{x,y}$, the set of all pairs $(i,j)$ such that $B$ gets $x_i$ in $\sigma_B(i,j)$ and $A$ does not get $x_j$ in $\sigma_B(i,j)$.
\item $T_y$, the set of all pairs $(i,j)$ such that $B$ does not get $x_i$ in $\sigma_B(i,j)$.
\end{enumerate}

We overload our notation and say that $x_\ell$ (resp. $y_{\ell}$) is also the value of $x_\ell$ (resp. $y_{\ell}$) for agent $A$. Based on Lemma~\ref{lm:z} and the definition of the types, we directly get the following claim.

\begin{cl}\label{cl:Txx_Txy_T_y}
Let $(i,j)$ be a pair with $0 \le i \le j \le n$.
\begin{enumerate}
    \item If $(i,j)\in T_{x,x}$ then $v_A(\sigma_A(i,j))+v_A(\sigma_B(i,j)) = x_i+x_j$ and $v_B(\sigma_A(i,j))+v_B(\sigma_B(i,j))\le x_j + x_i$.
    \item If $(i,j)\in T_{x,y}$ then $v_A(\sigma_A(i,j))+v_A(\sigma_B(i,j)) = x_i+y_j$ and $v_B(\sigma_A(i,j))+v_B(\sigma_B(i,j))\le y_j + x_i$.
    \item If $(i,j)\in T_{y}$ then $v_A(\sigma_A(i,j))+v_A(\sigma_B(i,j))\ge x_i+y_j$ and 
$v_B(\sigma_A(i,j))+v_B(\sigma_B(i,j))\le x_j + y_i$. 
\end{enumerate}
\end{cl}

The next lemma bounds $\er(\Omega(\sigma'))$ by bounding the contribution of each permutation in $\Omega(\sigma')$ according to its type.

\begin{lm}\label{lem:3-bound}
    \[\ er(\Omega(\sigma')) \leq \frac{\sum_{(i,j) \in T_{x,x}}(x_i+x_j)+\sum_{(i,j) \in T_{x,y}}(x_i+y_j) +\sum_{(i,j) \in T_{y}}(y_i+x_j)}{\sum_{(i,j) \in T_{x,x}}(x_i+x_j)+ \sum_{(i,j) \in T_{x,y}}(x_i+y_j)+ \sum_{(i,j) \in T_{y}}(x_i+y_j)}.\]
\end{lm}

\begin{proof}
The idea is to partition $\Omega(\sigma')$ according to the types of its permutations and then apply the bounds from Claim~\ref{cl:Txx_Txy_T_y}. By the definition of the index of a permutation and of the types, we obtain the following partition: 

\[\Omega(\sigma') = \bigsqcup_{(i,j) \in T_{x,x}\sqcup T_{x,y}\sqcup T_y} \{\sigma_A(i,j),\sigma_B(i,j)\}.\]

Therefore, 
\begin{align*}
\er(\Omega(\sigma')) =  \frac{\mathbb{E}[v_B(\hat\sigma)|\hat\sigma\in \Omega(\sigma')]}{\mathbb{E}[v_A(\hat\sigma)|\hat\sigma\in \Omega(\sigma')]}
&= \frac{\sum_{\sigma \in \Omega(\sigma')}v_B(\sigma)}{\sum_{\sigma \in \Omega(\sigma')}v_A(\sigma)}\\
&= \frac{\sum_{0 \leq i \leq j \leq n}
\bigl(v_B(\sigma_A(i,j)) + v_B(\sigma_B(i,j))\bigr)}
{\sum_{0 \leq i \leq j \leq n}
\bigl(v_A(\sigma_A(i,j)) + v_A(\sigma_B(i,j))\bigr)}\\
&= \frac{\sum_{(i,j)\in T_{x,x}\sqcup T_{x,y}\sqcup T_y}
\bigl(v_B(\sigma_A(i,j)) + v_B(\sigma_B(i,j))\bigr)}
{\sum_{(i,j)\in T_{x,x}\sqcup T_{x,y}\sqcup T_y}
\bigl(v_A(\sigma_A(i,j)) + v_A(\sigma_B(i,j))\bigr)}
\end{align*}
Claim~\ref{cl:Txx_Txy_T_y} provides an upper bound on $v_B(\sigma_A(i,j)) + v_B(\sigma_B(i,j))$ and a lower bound on $v_A(\sigma_A(i,j)) + v_A(\sigma_B(i,j))$, depending on the type of $(i,j)$. Combining these bounds with the previous equality yields the desired result.
 \[\ er(\Omega(\sigma')) \leq \frac{\sum_{(i,j) \in T_{x,x}}(x_i+x_j)+\sum_{(i,j) \in T_{x,y}}(x_i+y_j) +\sum_{(i,j) \in T_{y}}(y_i+x_j)}{\sum_{(i,j) \in T_{x,x}}(x_i+x_j)+ \sum_{(i,j) \in T_{x,y}}(x_i+y_j)+ \sum_{(i,j) \in T_{y}}(x_i+y_j)}. \qedhere\]
\end{proof}

Recall that, by construction, \(\mathbf{x} = (x_0, \dots, x_n)\) and \(\mathbf{y} = (y_0, \dots, y_n)\) are non-increasing sequences and \(x_\ell\ge y_\ell\ge 0\) for all \(\ell\). 

We define the polyhedral cone $K$ of all the pairs $(\mathbf{x},\mathbf{y}) \in \mathbb{R}^{n+1} \times \mathbb{R}^{n+1}$ satisfying the conditions:
\begin{itemize}
    \item $(\mathbf{x},\mathbf{y}) \neq (0,0)$
    \item $x_0 \geq x_1 \geq \dots \geq x_n \geq 0,$
    \item $y_0 \geq y_1 \geq \dots \geq y_n \geq 0,$
    \item $x_\ell \geq y_\ell$ for all $\ell$ with $0 \leq \ell \leq n$.
\end{itemize}

First, observe that any instance of the house allocation problem corresponds to a pair $(\mathbf{x},\mathbf{y}) \in K$ for a fixed $A$, $B$ and fixed $\sigma' \in \Omega'$. Next, we define the following function on $K$.

\[\mathcal{U}(\mathbf{x},\mathbf{y})
:= 1+\frac{\sum_{0\le i\leq j\le n}\max\{0,\Delta_j-\Delta_i\}}
{\sum_{0\leq i\leq j\leq n}(x_i+y_j)},\] 
where $\Delta_\ell := x_\ell-y_\ell \geq 0$. The next lemma bounds $\er(\Omega(\sigma'))$ by $\mathcal{U}(\mathbf{x},\mathbf{y})$.

\begin{lm}\label{lem:erub}
$er(\Omega(\sigma'))\leq \mathcal{U}(\mathbf{x},\mathbf{y})$.
\end{lm}
\begin{proof}

Define $$ 
\lambda = \sum_{(i,j) \in T_{x,x}} (x_i + x_j) + \sum_{(i,j) \in T_{x,y}} (x_i + y_j),\;
\tau_A = \sum_{(i,j) \in T_y} (x_i + y_j),\;
\text{and }\tau_B = \sum_{(i,j) \in T_y} (y_i + x_j).$$

Then Lemma~\ref{lem:3-bound} can be restated as 

\[
\er(\Omega(\sigma')) 
\leq \frac{\lambda+\tau_B}{\lambda+\tau_A}=1+\frac{\tau_B-\tau_A}{\lambda+\tau_A}.\]
We first bound the numerator,
\begin{align*}
    \tau_B-\tau_A
    &=\sum_{(i,j) \in T_y} (y_i + x_j) - \sum_{(i,j) \in T_y} (x_i + y_j)\\
    &= \sum_{(i,j) \in T_y} [(x_j - y_j) - (x_i - y_i)]\\
    &=\sum_{(i,j) \in T_y} (\Delta_j - \Delta_i).\\
    & \leq \sum_{(i,j) \in T_y} \max \set{0,\Delta_j - \Delta_i}\\
    & \leq \sum_{0\leq i \leq j \leq n} \max \set{0,\Delta_j - \Delta_i}.
\end{align*}
For the first inequality, we remove the negative terms; for the second, we add nonnegative terms.

We now turn to bounding the denominator, 
\begin{align*}
    \lambda+\tau_A 
    &= \sum_{(i,j) \in T_{x,x}} (x_i + x_j) + \sum_{(i,j) \in T_{x,y}} (x_i + y_j)+\sum_{(i,j) \in T_y} (x_i + y_j)\\
    &\geq \sum_{0\leq i\leq j\leq n}(x_i+y_j).
\end{align*}
The inequality follows from the fact that $x_j \ge y_j$, and from the observation that $T_{x,x}$, $T_{x,y}$, and $T_y$ form a partition of the set of ordered pairs $\set{(i,j)\mid 0 \le i \le j \le n}$.

Putting everything together, we obtain the desired result.
\begin{align*}
\er(\Omega(\sigma')) &\leq  1+\frac{\sum_{0\leq i \leq j \leq n} \max \set{0,\Delta_j - \Delta_i}}{\sum_{0\leq i\leq j\leq n}(x_i+y_j)} = \mathcal{U}(\mathbf{x},\mathbf{y}). \qedhere
\end{align*}
\end{proof}

We will now prove that \[\max_{(\mathbf{x},\mathbf{y}) \in K}\mathcal{U}(\mathbf{x},\mathbf{y}) \leq \sqrt{2}.\]

We first show, using a convexity argument, that the function is maximized at one of the step functions $(\mathbf{1}^{k_x}, \mathbf{1}^{k_y})$, where
$n+1 \ge k_x \ge k_y \ge 0$, $k_x \ge 1$, and $\mathbf{1}^k$ denotes the vector $(1,\ldots,1,0,\ldots,0)$ with $k$ ones followed by $n+1-k$ zeros.

\begin{restatable}{lm}{lmmaxstep} \label{lm:max_on_step_rays}
Let $N:K \rightarrow \mathbb{R}$ be a convex function that is homogeneous with degree $1$.
Let $D:K \rightarrow \mathbb{R}_{> 0}$, be a linear function that is strictly positive on $K$. Then

\[\max_{(\mathbf{x},\mathbf{y}) \in K}\frac{N(\mathbf{x},\mathbf{y})}{D(\mathbf{x},\mathbf{y})} = \max_{\substack{n+1 \ge k_x \ge k_y \ge 0 \\ k_x \ge 1}} \frac{N(\mathbf{1}^{k_x},\mathbf{1}^{k_y})}{D(\mathbf{1}^{k_x},\mathbf{1}^{k_y})}.\]
\end{restatable}

\begin{proof}
First, for all $(k_x,k_y)$ satisfying $k_x \geq 1$ and $n+1 \ge k_x \ge k_y \ge 0$, we have $(\mathbf{1}^{k_x},\mathbf{1}^{k_y}) \in K$ and therefore  
\[\frac{N(\mathbf{1}^{k_x},\mathbf{1}^{k_y})}{D(\mathbf{1}^{k_x},\mathbf{1}^{k_y})} \leq \max_{(\mathbf{x},\mathbf{y}) \in K}\frac{N(\mathbf{x},\mathbf{y})}{D(\mathbf{x},\mathbf{y})},\]

Now consider the other direction.
Fix $(\mathbf{x},\mathbf{y}) \in K$, and list the elements of the set
$\{x_0,x_1,\dots,x_n\} \cup \{y_0,y_1,\dots,y_n\} \cup \{0\}$ in decreasing order as
$t_0 > \cdots > t_m = 0$.
For each $j \in \{0,\dots,m-1\}$, define the integers
\[
k_x(j):=\bigl|\{i:\ x_i\ge t_j\}\bigr|,\qquad
k_y(j):=\bigl|\{i:\ y_i\ge t_j\}\bigr|.
\]
Because $\mathbf{x}$ and $\mathbf{y}$ are nonincreasing, the sets $\{ i : x_i \ge t_j \}$ and $\{ i : y_i \ge t_j \}$ are initial segments of the form $\{0,\dots,k_x(j)-1\}$ and $\{0,\dots,k_y(j)-1\}$, respectively.
By definition, $k_x$ and $k_y$ are nondecreasing functions of $j$. Moreover, since $\mathbf{x} \ge \mathbf{y}$, we have $k_x(j) \ge k_y(j)$ for all $j \in \{0,\dots,m-1\}$.

Set $\alpha_j:=t_j-t_{j+1}> 0$. We claim that 
\begin{equation}\label{eq:decomp}
(\mathbf{x},\mathbf{y})=\sum_{j=0}^{m-1} \alpha_j\bigl(\mathbf{1}^{k_x(j)},\mathbf{1}^{k_y(j)}\bigr).
\end{equation}

To prove this equation, for all $\ell\in \{0,\dots, n\}$, set $j_\ell$ such that $x_\ell = t_{j_\ell}$. Then, 
\begin{align*}
    \sum_{j=0}^{m-1} \alpha_j (\mathbf{1}^{k_x(j)})_\ell &= \sum_{j \geq j_{\ell}} \alpha_j\\
    & = \sum_{j \geq j_{\ell}} (t_j - t_{j+1})\\
    & = t_{j_\ell} - t_m\\
    &=  x_\ell.
\end{align*}
A similar telescopic sum proves the equality for $y_\ell$. 

Note that because $k_x(j) \geq k_y(j)$, and $k_x(j) \geq 1$ for all $j \in \{0,\dots, m-1\}$, all the pairs $(\mathbf{1}^{k_x(j)},\mathbf{1}^{k_y(j)}\bigr)$ appearing in Equation~(\ref{eq:decomp}) are in $K$. 

We now show that $N$ is subadditive. For any $z_1,z_2\in K$, convexity and homogeneity imply subadditivity because
\[
N(z_1+z_2)=2\,N\!\left(\frac{z_1+z_2}{2}\right)\le 2 \left(\frac{N(z_1)+N(z_2)}{2}\right) = N(z_1) + N(z_2).
\]
By induction, for any $m\in\mathbb{N}$ and any $z_1,\dots,z_m\in K$,
\[N\!\left(\sum_{j=1}^m z_j\right)\le \sum_{j=1}^m N(z_j).\]

Therefore 
\begin{align*}
 N(\mathbf{x},\mathbf{y}) &= N\big(\sum_{j=0}^{m-1} \alpha_j\bigl(\mathbf{1}^{k_x(j)},\mathbf{1}^{k_y(j)}\bigr)\big) \\ 
 &\leq \sum_{j = 0}^{m-1} N(\alpha_j(\mathbf{1}^{k_x(j)},\mathbf{1}^{k_y(j)}))\\
 &= \sum_{j = 0}^{m-1} \alpha_j N(\mathbf{1}^{k_x(j)},\mathbf{1}^{k_y(j)}).
\end{align*}
Here the last step uses homogeneity again. Using the linearity of $D$, we get
\begin{align*}
    D(\mathbf{x},\mathbf{y}) &= D\big(\sum_{j=0}^{m-1} \alpha_j(\mathbf{1}^{k_x(j)},\mathbf{1}^{k_y(j)})\big)\\
    &=  \sum_{j=0}^{m-1} \alpha_j D(\mathbf{1}^{k_x(j)},\mathbf{1}^{k_y(j)}).
\end{align*}
Moreover for all $j$, $\alpha_j > 0$ and because $(\mathbf{1}^{k_x(j)},\mathbf{1}^{k_y(j)}) \in K$, $D(\mathbf{1}^{k_x(j)},\mathbf{1}^{k_y(j)})> 0 $. Using Observation~\ref{obs:fraction}, we then obtain
\begin{align*}
    \frac{N(\mathbf{x},\mathbf{y})}{D(\mathbf{x},\mathbf{y})} &\leq \frac{\sum_{j = 0}^{m-1} \alpha_j N(\mathbf{1}^{k_x(j)},\mathbf{1}^{k_y(j)})}{\sum_{j=0}^{m-1} \alpha_j D(\mathbf{1}^{k_x(j)},\mathbf{1}^{k_y(j)})}\\
    & \leq \max_{j\in \{0,\dots, m-1\}} \frac{ N(\mathbf{1}^{k_x(j)},\mathbf{1}^{k_y(j)})}{ D(\mathbf{1}^{k_x(j)},\mathbf{1}^{k_y(j)})}.\\
    & \leq \max_{\substack{n+1 \ge k_x \ge k_y \ge 0 \\ k_x \ge 1}} \frac{N(\mathbf{1}^{k_x},\mathbf{1}^{k_y})}{D(\mathbf{1}^{k_x},\mathbf{1}^{k_y})}.
\end{align*}
This concludes the proof.
\end{proof}

\begin{proof}[Proof of Theorem~\ref{thm:upper}]
    We are now ready to prove Lemma~\ref{lm:c'upper_bound} and, by implication, Theorem~\ref{thm:upper}. Combining Lemmas~\ref{lm:order} and~\ref{lem:erub}, it suffices to prove that 
    \[\max_{(\mathbf{x},\mathbf{y}) \in K}\mathcal{U}(\mathbf{x},\mathbf{y}) \leq \sqrt{2}.\]

    We apply Lemma~\ref{lm:max_on_step_rays}, with $N(\mathbf{x},\mathbf{y}):= \sum_{0\leq i \leq  j \leq n} \max \set{0,\Delta_j - \Delta_i}$ where $\Delta_\ell = x_\ell - y_\ell$ for all $\ell \in \{0, \dots,n\}$, and $D(\mathbf{x},\mathbf{y}):= \sum_{0\leq i\leq j\leq n}(x_i+y_j)$. $D$ is linear and strictly positive on $K$ because we assumed that $A$ has a positive value for at least one item, and therefore $x_0>0$. For $N$, each term $\max\{0,\Delta_j-\Delta_i\}$ is the maximum of $0$ and an affine function and is therefore convex and positively $1$-homogeneous. Consequently, their sum $N(\mathbf{x},\mathbf{y})$ is convex and $1$-homogeneous. The conditions of Lemma~\ref{lm:max_on_step_rays} are satisfied with $\mathcal{U} = 1 + \frac{N}{D}$, therefore 
    \[\max_{(\mathbf{x},\mathbf{y}) \in K}\mathcal{U}(\mathbf{x},\mathbf{y})=  \max_{\substack{n+1 \ge k_x \ge k_y \ge 0 \\ k_x \ge 1}}\mathcal{U}(\mathbf{1}^{k_x},\mathbf{1}^{k_y}).\]
    
    We will now compute $\mathcal{U}(\mathbf{1}^{k_x},\mathbf{1}^{k_y})$ for $n+1\geq k_x\geq k_y\geq 0$ and $k_x\geq 1$. In this case 
    \[
    \Delta_\ell=\mathbf{1}^{k_x}_\ell-\mathbf{1}^{k_y}_\ell=
    \begin{cases}
    0, & \ell<k_y,\\
    1, & k_y\le \ell<k_x,\\
    0, & \ell\ge k_x,
    \end{cases}
    \]
    so $\max\{0,\Delta_j-\Delta_i\}=1$ if and only if $\Delta_i=0$ and $\Delta_j=1$, that is, when $i\in\{0,\dots,k_y-1\}$ and $j\in\{k_y,\dots,k_x-1\}$ with $i\le j$. Hence
    \begin{align*}
    N(\mathbf{1}^{k_x},\mathbf{1}^{k_y})
    &=\sum_{j=k_y}^{k_x-1}\#\{\,i\le j:\Delta_i=0\,\}\\
    &=\sum_{j=k_y}^{k_x-1} k_y\\
    &= k_y(k_x-k_y).
    \end{align*}
    Moreover,
    \begin{align*}
        D(\mathbf{1}^{k_x},\mathbf{1}^{k_y})
    &=\sum_{0\le i\le j\le n}(x_i+y_j)\\
    &=\sum_{i=0}^n x_i(n+1-i)+\sum_{j=0}^n y_j(j+1)\\
    &=\frac{k_x(2n+3-k_x)}{2}+\frac{k_y(k_y+1)}{2}.
    \end{align*}

    Therefore
    \[
    \mathcal{U}(\mathbf{1}^{k_x},\mathbf{1}^{k_y})
    =1+\frac{2k_y(k_x-k_y)}{k_x(2n+3-k_x)+k_y(k_y+1)}.
    \]
    Clearly, for fixed $k_x,k_y$, decreasing $n$ decreases the denominator and hence increases the ratio, so the maximum over admissible $n$ is attained at the smallest possible value, namely $n=k_x-1$. In that case,
    \[
    k_x(2n+3-k_x)=k_x(2(k_x-1)+3-k_x)=k_x(k_x+1),
    \]
    and therefore
    \[
    \mathcal{U}(\mathbf{1}^{k_x},\mathbf{1}^{k_y})
    \le 1+\frac{2k_y(k_x-k_y)}{k_x(k_x+1)+k_y(k_y+1)}.
    \]
    Let $p=\frac{k_y}{k_x}\in[0,1]$.
    \begin{align*}
        \frac{2k_y(k_x-k_y)}{k_x(k_x+1)+k_y(k_y+1)} &=\frac{2p(1-p)k_x^2}{(1+p^2)k_x^2+(1+p)k_x}\\
    &=\frac{2p(1-p)k_x}{(1+p^2)k_x+(1+p)}\\
    &\le \frac{2p(1-p)}{1+p^2},
    \end{align*}
    where the last inequality comes from $(1+p^2)k_x+(1+p)\ge (1+p^2)k_x$ and $k_x\ge 1$. 
    
    It remains to maximize $g(p):=\frac{2p(1-p)}{1+p^2}$ over $p\in[0,1]$. Observe that
    \[
    g(p)=\frac{2p(1-p)}{1+p^2}\le M
    \iff
    M(1+p^2)-2p(1-p)\ge0.
    \]
    Take \(M=\sqrt2-1\). Then
    \[
    (\sqrt2-1)(1+p^2)-2p(1-p)
    =(\sqrt2+1)\bigl(p-(\sqrt2-1)\bigr)^2\ge 0.
\]
Hence \(g(p)\le\sqrt2-1\) on \([0,1]\), with equality at \(p=\sqrt2-1\), for which the envy-ratio is at most
    \[1+\frac{2p(1-p)}{1+p^2} = \sqrt{2}. \qedhere\] 
\end{proof}

\subsection{A $\sqrt{2}$ Lower Bound on the Envy-Ratio}\label{sec:lower-single}

\citet*{feldman2024breaking} presented a sequence of instances showing that the envy-ratio of RSD is at least $\sqrt{2}$. For completeness, we present a similar result here with a slightly simpler construction, adapted from the sequence of examples in \citet*{feldman2024breaking}.

\begin{restatable}{thm}{lowerfeldman}\label{thm:lower}
\textnormal{\cite{feldman2024breaking}}%
\quad For any $\varepsilon>0$, there exists an instance for which $\er(\Omega) \ge \sqrt{2}-\varepsilon.$
\end{restatable}

\begin{proof}
Consider the following instance of the house allocation problem, consisting of agents $A$ and $B$, and a set $\mathcal{C}$ of $n$ other agents $C_1,\ldots,C_n$. The valuations of the agents are given in Table~\ref{tab:lb-sqrt2}. This instance resembles the lower bound instance of Table~\ref{tab:lb-54} presented in the Preliminaries section, except that it contains $k$ copies of items of type~$a$, $n-k+1$ copies of items of type~$c$, and $n$ copies of agent~$C$.
\begin{table}[ht]
    \centering
    \begin{tabularx}{0.65\textwidth}{ c | c c c | c | c c c }
           & \multicolumn{7}{c}{{\bf Item}} \\
           \cline{2-8}
        {\bf Agent} & $a_1$ & $\ldots$ & $a_k$ & $b$ & $c_1$ & $\ldots$ & $c_{n-k+1}$ \\
        \cline{1-8}
        {\tt A} & $1$ & $1$ & $1$ & $1-\varepsilon$ & $0$ & $0$ & $0$ \\
        {\tt B} & $1-\varepsilon$ & $1-\varepsilon$ & $1-\varepsilon$ & $1$ & $0$ & $0$ & $0$ \\
        {\tt C$_1$} & $1$ & $1$ & $1$ & $0$ & $1-\varepsilon$ & $1-\varepsilon$ & $1-\varepsilon$ \\
        {$\vdots$} &  & $\vdots$ & $ $ & $\vdots$ & $ $ & $\vdots$ & $ $ \\
        {\tt C$_n$} & $1$ & $1$ & $1$ & $0$ & $1-\varepsilon$ & $1-\varepsilon$ & $1-\varepsilon$ \\
    \end{tabularx}
    \caption{Valuations for RSD instance demonstrating a factor-$\sqrt{2}$ lower bound on the envy-ratio}
    \label{tab:lb-sqrt2}
\end{table}

We proceed with a case analysis based upon the relative order of arrival of agents~$A$, $B$, and the $k^{\text{th}}$ agent to arrive from the set $\mathcal{C}$, which we call agent~$C^*$. The six possible cases, along with their approximate probabilities of occurrence when both $n$ and $k$ are large, are listed in Table~\ref{tab:lb-sqrt2-cases}.

\begin{table}[ht]
    \centering
    \begin{tabularx}{0.82\textwidth}{c|ccc|c}
           & \multicolumn{3}{c|}{{\bf Agent}} & \\
           \cline{2-4}
        {\bf Relative Ordering} & {\tt A} & {\tt B} & {\tt C$^*$} & {\bf Probability} \\
        \cline{1-5}
        {\tt ABC$^*$} & type-$a$ item & $b$ & type-$c$ item & $\approx \frac{k^2}{2n^2}$ \\
        {\tt AC$^*$B} & type-$a$ item & $b$ & type-$c$ item & $\approx \frac{k}{n}(\frac{n-k}{n})$ \\
        {\tt BAC$^*$} & type-$a$ item & $b$ & type-$c$ item & $\approx \frac{k^2}{2n^2}$ \\
        {\tt BC$^*$A} & type-$c$ item & $b$ & type-$a$ item & $\approx \frac{k}{n}(\frac{n-k}{n})$ \\
        {\tt C$^*$AB} & $b$ & type-$c$ item & type-$a$ item & $\approx \frac{(n-k)^2}{2n^2}$ \\
        {\tt C$^*$BA} & type-$c$ item & $b$ & type-$a$ item & $\approx \frac{(n-k)^2}{2n^2}$ 
    \end{tabularx}
    \caption{Cases for RSD instance demonstrating a factor-$\sqrt{2}$ lower bound on the envy-ratio}
    \label{tab:lb-sqrt2-cases}
\end{table}

Let $X_A$ and $X_B$ be the random items obtained by agents $A$ and $B$ respectively. For $\varepsilon\rightarrow0$, the envy-ratio is bounded below by
\[
\frac{\mathbb{E}[v(X_B)]}{\mathbb{E}[v(X_A)]} \approx\frac{2\cdot\frac{k^2}{2n^2} + 2\cdot\frac{k(n-k)}{n^2} + \frac{(n-k)^2}{2n^2}}{2\cdot\frac{k^2}{2n^2} + \frac{k(n-k)}{n^2} + \frac{(n-k)^2}{2n^2}} 
= \frac{n^2-k^2+2kn}{n^2+k^2},
\]
which for large $n$ and $k\approx(\sqrt{2}-1)\cdot n$ reaches a maximum value approaching $\sqrt{2}$.
\end{proof}

Together, Theorems~\ref{thm:upper} and~\ref{thm:lower} give us our headline result (Theorem~\ref{thm:main}), which is that the envy-ratio of random serial dictatorship is exactly $\sqrt{2}$ in the house allocation problem.

\section{Beyond One Round: Additive Valuations}\label{sec:additive}

In this section, and in Section~\ref{sec:nonadditive},
we are interested in the more general setting in which the number of objects may exceed the number of agents, and the agents themselves have more general valuation functions rather than unit-demand functions. To model this setting, each agent is equipped with a combinatorial valuation function that assigns a value for every subset of items. The valuation functions are normalized (i.e., $v_i(\varnothing) = 0$) and monotone (i.e., nondecreasing with respect to set inclusion). Our results in Section~\ref{sec:additive} are for the case of additive valuations: the set function $v$ is {\em additive} if $v(S) = \sum_{j\in S}v(\{j\})$ for all $S\subseteq M$. In Section~\ref{sec:nonadditive}, we analyze more general valuation classes, namely submodular, XOS, and subadditive valuations. A set function $v$ is:
\begin{itemize}
    \item \emph{submodular}, if for any pair $S,T$ of item sets with $S\subseteq T\subseteq M$ and any item $j\in M$, $v(S\cup\{j\})-v(S) \geq v(T\cup\{j\})-v(T)$, i.e., every item has decreasing marginal values.
    \item \emph{XOS}, or \emph{fractionally subadditive}, if there exists a nonempty collection $v_1,\ldots,v_k$ of additive set functions such that for any $S\subseteq M$, $v(S) = \max_{i\in[k]}v_i(S)$.
    \item \emph{subadditive}, or \emph{complement-free}, if for any pair $S,T$ of item sets, $v(S) + v(T) \geq v(S\cup T)$.
\end{itemize}

\citet*{lehmann2001combinatorial} showed that these classes form the following hierarchy.
\begin{center}
    \tt additive $\subsetneq$ submodular $\subsetneq$ XOS $\subsetneq$ subadditive
\end{center}

Given an agent $A$, a set of items $S_A$, and a set of unallocated items $R$, the most-preferred remaining item of $A$ is the item that maximizes its \emph{marginal value} $v(S_A \cup {j}) - v(S_A)$ among all items $j \in R$. As in the RSD case, we assume that $A$ has a total order over $R$ (which may depend on $S_A$) that is coherent with the order induced by these marginal values, so that ties are broken according to a fixed rule and each agent has a unique most-preferred item among any subset of remaining items.

The two mechanisms that we analyze are the natural extensions of RSD to these more general settings. Both mechanisms assign objects in rounds until no objects remain. To simplify the exposition, we assume throughout that $|M|$ is a multiple of $|N|$, i.e. $|M|=r|N|$ and the mechanism terminates after exactly $r$ rounds\footnote{As in the RSD case, if $|M|$ is not a multiple of $|N|$, we can add dummy items that always contribute a marginal value of zero, and this does not affect any envy-ratio. All classes of valuation functions we consider—additive, submodular, XOS, and subadditive—are stable under the addition of such dummy items.}.

\paragraph{Iterated Random Serial Dictatorship.}
In the iterated-RSD mechanism, the items are allocated with $r$ iterations of the RSD mechanism. At the start of each round, a new permutation is drawn uniformly to allocate one item per agent. More precisely:
\begin{enumerate}[topsep=2pt,itemsep=0pt,parsep=2pt]
    \item While there remains an unallocated item:
    \begin{enumerate}[label=(\roman*),topsep=2pt,itemsep=0pt,parsep=2pt]
        \item Draw a permutation $\hat\sigma$ uniformly at random from $\Omega$.
        \item For $k = 1, \dots, |N|$, assign to agent $\hat\sigma(k)$ its most-preferred remaining item.
    \end{enumerate}
\end{enumerate}

\paragraph{Randomized Round-Robin.}
In the randomized round-robin mechanism~\cite{freeman2020best}, a single permutation is selected and remains fixed throughout the item allocation process. It operates as follows:
\begin{enumerate}[topsep=2pt,itemsep=0pt,parsep=2pt]
    \item Draw a permutation $\hat\sigma$ uniformly at random from $\Omega$.
    \item While there remains an unallocated item:
    \begin{enumerate}[label=(\roman*),topsep=2pt,itemsep=0pt,parsep=2pt]
    \item For $k = 1, \dots, |N|$, assign agent $\hat\sigma(k)$ its most-preferred remaining item.
    \end{enumerate}
\end{enumerate}

Since the randomized round-robin and iterated-RSD mechanisms select an item myopically, i.e., with the highest marginal value among remaining items, they naturally generalize to settings in which the agents do not have additive valuations. In particular, these mechanisms remain well-defined for submodular and subadditive valuations, allowing for a systematic study of the envy-ratio in these more general settings, which we do in Section~\ref{sec:nonadditive}.


In this section, for additive valuations, we 
bound the envy-ratio of both mechanisms. We first show 
that, for the iterated-RSD mechanism, the envy-ratio remains equal to $\sqrt{2}$ when the agents have additive valuations.

\thmirsdadd*

The lower bound of $\sqrt{2}$ follows from Theorem~\ref{thm:lower}, since any one-round instance of RSD with $|M| = |N|$ can be viewed as an additive instance of iterated-RSD.

For the upper bound, we use the following simple claim, which is a direct consequence of our main theorem, Theorem~\ref{thm:main}, and the fact that the sub-instance corresponding to the $k^\text{th}$ round is equivalent to an independent instance of RSD in a house allocation problem with unit demands.

\begin{cl}\label{cl:iRSD}
    During the execution of the iterated-RSD mechanism, for all $k$ with $1\leq k\leq r$, let $a(k)$ (resp. $b(k)$) be the random item that $A$ selects during round $k$. Then
    $\bb{E}[v(b(k))] \leq \sqrt{2} \cdot \bb{E}[v(a(k))]$.
\end{cl}



We are now ready to prove the upper bound for the additive case.

\begin{proof}[Proof of Theorem~\ref{thm:irsdadd}.]
As in Claim~\ref{cl:iRSD}, for all $k$ with $1\leq k\leq r$, let $a(k)$ (resp. $b(k)$) be the random item that $A$ selects during the $k^\text{th}$ round, and $S(k) = \{a(1), \dots , a(k-1)\}$ be the random bundle of $A$ at the start of the $k^\text{th}$ round. Moreover let $S_A$ (resp $S_B$) be the random bundle of $A$ (resp of $B$) when the mechanism terminates. We have
\begin{align*}
    \bb{E}[v(B)] = \sum_{k=1}^r \bb{E}[b(k)] \leq \sum_{k=1}^r \sqrt{2} \cdot \bb{E}[a(k)] = \sqrt{2} \cdot \bb[v(A)].\qquad\qedhere
\end{align*}
\end{proof}

However, the picture is not as simple for randomized round-robin. Since the permutation is fixed at the start of the first round and does not vary between rounds, the rounds are no longer independent. In fact, in contrast with iterated-RSD, the envy-ratio no longer remains at $\sqrt{2}$, but increases to at least $3/2$ for randomized round-robin. On the positive side, we present a constant-factor upper bound for this mechanism, showing that the envy-ratio is at most $1+(1/\sqrt{2})$.

\thmrrradd*

The lower bound follows from a relatively simple 2-round example with the following structure. Both~$A$ and~$B$ have a value of $0$ in the second round in nearly every permutation in which $A$ appears before $B$, but $B$ obtains an item of high value in the second round in every permutation in which it appears before $A$ (which happens half the time). This example is presented in the following proof.

\begin{thm}\label{thm:loweraddRRR}
For any $\varepsilon>0$, there exists an instance for which the envy-ratio of randomized round-robin is at least $ 3/2-\varepsilon.$
\end{thm}
\begin{proof}
The lower bound of $3/2$ follows from a relatively simple example consisting of agents $A$ and $B$ and $n$ identical agents $C_1,\ldots,C_n$, and of $2(n+2)$ items for a total of two rounds.
The valuations of the agents are given in Table~\ref{tab:lb-32}.

\begin{table}[ht]
    \centering
    \begin{tabularx}{0.65\textwidth}{ c | c | c | c | c c c | c c c }
        & \multicolumn{9}{c}{{\bf Item}} \\
        \cline{2-10}
        {\bf Agent} & $i_1$ & $i_2$ & $i_3$ & $d_1$ & $\ldots$ & $d_n$ & $e_1$ & $\ldots$ & $e_{n+1}$ \\
        \cline{1-10}
        {\tt A} & $1$ & $1-\varepsilon$ & $1-2\varepsilon$ & $0$ & $0$ & $0$ & $0$ & $0$ & $0$ \\
        {\tt B} & $1$ & $1-2\varepsilon$ & $1-\varepsilon$ & $0$ & $0$ & $0$ & $0$ & $0$ & $0$ \\
        {\tt C$_1$} & $0$ & $1-\varepsilon$ & $0$ & $1$ & $1$ & $1$ & $1-2\varepsilon$ & $1-2\varepsilon$ & $1-2\varepsilon$ \\
        {$\vdots$} & $0$ & $1-\varepsilon$ & $0$ & $1$ & $1$ & $1$ & $1-2\varepsilon$ & $1-2\varepsilon$ & $1-2\varepsilon$ \\
        {\tt C$_n$} & $0$ & $1-\varepsilon$ & $0$ & $1$ & $1$ & $1$ & $1-2\varepsilon$ & $1-2\varepsilon$ & $1-2\varepsilon$ \\
    \end{tabularx}
    \caption{Valuations for randomized round-robin instance demonstrating a factor-$(3/2)$ lower bound}
    \label{tab:lb-32}
\end{table}

Observe that, in any permutation, the items in $\{d_1,\ldots,d_n\}$ are assigned to $C_1,\ldots,C_n$ in the first round. Take any permutation in which $B$ goes before $A$. Then $B$ selects $i_1$ in the first round and $A$ selects $i_2$. In the second round, $B$ selects the item~$i_3$, which is always available, but $A$ is always assigned an item of zero value. Take any permutation in which $A$ goes before $B$.
Then $A$ selects $i_1$ in the first round and $B$ selects $i_3$. Furthermore, with high probability, 
the permutation chosen by randomized round-robin
has an agent from the set $\{C_1,\ldots,C_n\}$ appearing first. In that case, that agents selects the item~$i_2$ in the second round, before agent~$A$ and agent~$B$ select their second round items. This leads to both $A$ and $B$ receiving an item of zero value. Let $X_A$ and $X_B$ respectively denote the bundles obtained by agents~$A$ and~$B$. Then, the envy-ratio is at least
\[
\frac{\mathbb{E}[v(X_B)]}{\mathbb{E}[v(X_A)]} \approx \frac{\frac12(1+1-2\varepsilon)+\frac12(1-2\varepsilon)}{\frac12(1-\varepsilon)+\frac12(1)} = \frac{3-4\varepsilon}{2-\varepsilon},
\]
which for large $n$ and $\varepsilon\rightarrow0$ gives a lower bound of $3/2$.

\end{proof}

Next, we present our factor-$1+(1/\sqrt{2})$ upper bound on the envy-ratio of randomized round-robin.

\begin{restatable}{thm}{thmupperaddrrr}\label{thm:upperaddRRR}
    For additive valuations, the envy-ratio of randomized round-robin is at most $1+\frac{1}{\sqrt{2}}$.
\end{restatable}

Fix two distinct agents $A$ and $B$. Let $\Omega_A$ (resp. $\Omega_B$) be the set of permutations in which $A$ appears before $B$ (resp. $B$ appears before $A$). As before, we let $v$ denote the valuation of agent $A$, and let $\hat\sigma\in\Omega$ denote the permutation selected by randomized round-robin at the start of its execution. For all $k \in \{1, \dots, r\}$, let $a_\sigma(k)$ (resp. $b_\sigma(k)$) be the item that $A$ (resp. $B$) selects in round $k$.

First, we require the following simple observation.

\begin{obs}\label{obs:RR-late}
     For all $k \in \{1, \dots, r\}$, if $\sigma \in \Omega_A$ then $v(b_\sigma(i))\leq v(a_\sigma(i))$. For all $k \in \{1, \dots, r-1\}$, if $\sigma \in \Omega_B$ then $v(b_\sigma(i+1))\leq v(a_\sigma(i))$.
\end{obs}



     

The next lemma provides a lower bound on what $A$ receives in the first round of randomized round-robin, based on a combination of what $B$ receives in the first round and what $B$ receives in the second round when it appears before $A$ in the random permutation. Note that we consider only the values of individual items.

\begin{restatable}{lm}{RRearly}\label{lm:RR-early} 
     \begin{align*}
        \frac12\bb{E}[v(b_{\hat \sigma}(1))|\hat \sigma \in \Omega_A]\quad+\quad &\frac12\bb{E}[v(b_{\hat \sigma}(1))|{\hat \sigma}\in \Omega_B]\quad+\quad \frac12\bb{E}[v(b_{\hat \sigma}(2))|{\hat \sigma}\in \Omega_B]&\\
        \leq\quad &\left(1+\frac{1}{\sqrt{2}}\right)\bb{E}[v(a_{\hat \sigma}(1))|{\hat \sigma}\in \Omega].
     \end{align*}
\end{restatable}

\begin{proof}
First, we split the expected value of the item that $B$ selects in the first round into two parts, depending on the relative positions of $A$ and $B$.
\[
\frac12\bb{E}[v(b_{\hat\sigma}(1))\mid \hat \sigma \in\Omega_A]
+\frac12\bb{E}[v(b_{\hat\sigma}(1))\mid \hat \sigma \in\Omega_B]
=\bb{E}[v(b_{\hat\sigma}(1)) \hat \sigma \in \mid\Omega].
\]

Even when $B$ is first, the item that $A$ selects in the first round has a higher value than the item that $B$ selects in the second round. This corresponds to the second point of Observation~$\ref{obs:RR-late}$ for the first round. Taking the expectation over $\Omega_{B}$, we get
\[
\bb{E}[v(b_{\hat\sigma}(2))\mid \hat \sigma \in\Omega_B]
\le \bb{E}[v(a_{\hat\sigma}(1))\mid \hat \sigma \in\Omega_B].
\]

Combining the above, the left-hand side of the lemma is at most
\begin{equation}\label{eq:RR-early}
    \bb{E}[v(b_{\hat\sigma}(1))\mid \hat \sigma \in\Omega]
+\frac12\bb{E}[v(a_{\hat\sigma}(1))\mid \hat \sigma \in\Omega_B].
\end{equation}

We now give two bounds on this expression.

\paragraph{Bound 1.} The first bound comes from our main Theorem~\ref{thm:main} applied to the first round of the mechanism,
\[
\bb{E}[v(b_{\hat\sigma}(1))\mid \hat \sigma \in\Omega]
\le \sqrt{2}\,\bb{E}[v(a_{\hat\sigma}(1)) \mid \hat \sigma \in \Omega],
\]
which implies a first bound for expression~\ref{eq:RR-early}:
\begin{align*}
&\bb{E}[v(b_{\hat\sigma}(1))\mid \hat \sigma \in \Omega]
+\tfrac12\bb{E}[v(a_{\hat\sigma}(1))\mid  \hat \sigma \in \Omega_B]\\
\le\quad
&\tfrac{\sqrt{2}}{2}\bb{E}[v(a_{\hat\sigma}(1))\mid \hat \sigma \in\Omega_A]
+\tfrac{1+\sqrt{2}}{2}\bb{E}[v(a_{\hat\sigma}(1))\mid \hat \sigma \in \Omega_B].
\end{align*}

\paragraph{Bound 2.}
The second bound comes from pairing each permutation $\sigma_A\in\Omega_A$ with $\sigma_B\in\Omega_B$ obtained by swapping the positions of $A$ and $B$. Then $v(a_1(\sigma_A)) \ge v(b_1(\sigma_B))$.

Taking the expectation over $\Omega_{A}$ and $\Omega_{B}$ respectively, we get
\[\bb{E}[v(b_{\hat \sigma}(1))|{\hat \sigma}\in \Omega_B] \leq \bb{E}[v(a_{\hat \sigma}(1))|{\hat \sigma}\in \Omega_A].\]

This gives us a second bound for expression~\ref{eq:RR-early}.
\[
\bb{E}[v(b_{\hat\sigma}(1))\mid \hat \sigma \in\Omega]
+\tfrac12\bb{E}[v(a_{\hat\sigma}(1))\mid \hat \sigma \in \Omega_B]
\le
\bb{E}[v(a_{\hat\sigma}(1))\mid \hat \sigma \in \Omega_A]
+\tfrac12\bb{E}[v(a_{\hat\sigma}(1))\mid \hat \sigma \in \Omega_B].
\]

We now combine the two bounds and optimize the ratio. Let
\[
\alpha := \bb{E}[v(a_{\hat\sigma}(1))\mid \hat \sigma \in \Omega_A],
\qquad
\beta := \bb{E}[v(a_{\hat\sigma}(1))\mid \hat \sigma \in \Omega_B].
\]

Combining the two bounds yields
\[
\bb{E}[v(b_{\hat\sigma}(1))\mid \hat \sigma \in \Omega]
+\tfrac12\bb{E}[v(a_{\hat\sigma}(1))\mid \hat \sigma \in \Omega_B]
\le
\min\left\{
\tfrac{\sqrt{2}}{2}\alpha+\tfrac{1+\sqrt{2}}{2}\beta,\;
\alpha+\tfrac12\beta
\right\},
\]

where
$\alpha := \bb{E}[v(a_{\hat\sigma}(1))\mid \hat \sigma \in \Omega_A]$, and $\beta := \bb{E}[v(a_{\hat\sigma}(1))\mid \hat \sigma \in \Omega_B]$.

On the other hand, $\bb{E}[v(a_{\hat\sigma}(1))\mid \hat \sigma \in \Omega]=\tfrac12(\alpha+\beta).$

To bound the ratio of the two expressions, we maximize $\frac{\min\{\frac{\sqrt{2}}{2}\alpha+\frac{1+\sqrt{2}}{2}\beta,\;
\alpha+\frac12\beta\}}{\frac12(\alpha+\beta)}$ over $\alpha,\beta\ge0$. 
The expected value $A$ gets when picking before $B$ is at least the expected value when $B$ picks before $A$, so $\alpha\geq \beta$.
By homogeneity, assume $\alpha+\beta=2$, i.e., $1\le\alpha\le2$. 
Substituting $\beta=2-\alpha$ gives
\[
\max_{1\le\alpha\le2}
\min\left\{
1+\sqrt{2}-\frac{\alpha}{2},\;
1+\frac{\alpha}{2}
\right\}
=
1+\max_{1\le\alpha\le2}
\min\left\{
\sqrt{2}-\frac{\alpha}{2},\;
\frac{\alpha}{2}
\right\}.
\]
The maximum is attained at $\alpha=\sqrt{2}$, yielding the value $1+\frac{1}{\sqrt{2}}$.
\end{proof}

Using the above lemma, we now derive the upper bound.

\begin{proof}[Proof of Theorem~\ref{thm:upperaddRRR}]
    Let $\hat \sigma $ be the random permutation selected by the mechanism. For the permutations where $A$ comes before $B$, Observation~$\ref{obs:RR-late}$ upper-bounds the value of the item that $B$ gets in any round $i$ by the value of the item that $A$ gets in round $i$. Summing over all rounds $i \ge 2$, and taking the expectation over $\Omega_A$, we obtain \[\sum_{i=2}^r \bb{E}[v(b_{\hat{\sigma}}(i))|\hat{\sigma}\in \Omega_A] \le \sum_{i=2}^r \bb{E}[v(a_{\hat{\sigma}}(i))|\hat{\sigma}\in \Omega_A].\]

    For the permutations where $B$ comes before $A$, Observation~$\ref{obs:RR-late}$ upper-bounds the value of the item that $B$ gets in any round $i+1$  by the value of the item that $A$ gets in round $i$. Summing over all rounds $i \ge 2$, and taking expectation over $\Omega_B$ we have \[\sum_{i=2}^{r-1} \bb{E}[v(b_{\hat{\sigma}}(i+1))|\hat{\sigma}\in \Omega_A]\leq \sum_{i=2}^{r-1} \bb{E}[v(a_{\hat{\sigma}}(i))|\hat{\sigma}\in \Omega_A].\]

    Additionally, by Lemma~\ref{lm:RR-early} we get:
    \begin{align*}
        \frac12\bb{E}[v(b_{\hat \sigma}(1))|\hat \sigma \in \Omega_A]\quad+\quad &\frac12\bb{E}[v(b_{\hat \sigma}(1))|{\hat \sigma}\in \Omega_B]\quad+\quad \frac12\bb{E}[v(b_{\hat \sigma}(2))|{\hat \sigma}\in \Omega_B]&\\
        \leq\quad &\left(1+\frac{1}{\sqrt{2}}\right)\bb{E}[v(a_{\hat \sigma}(1))|{\hat \sigma}\in \Omega].
     \end{align*}

    Combining the above, we get
    \begin{align*}
        \bb{E}[v_B(\hat{\sigma})|\hat{\sigma}\in \Omega]
        &=\sum_{i=1}^r \bb{E}[v(b_{\hat{\sigma}}(i))|\hat{\sigma}\in \Omega]\\
        &=\frac{1}{2}\sum_{i=1}^r \bb{E}[v(b_{\hat{\sigma}}(i))|\hat{\sigma}\in \Omega_A]+\frac{1}{2}\sum_{i=1}^r \bb{E}[v(b_{\hat{\sigma}}(i))|\hat{\sigma}\in \Omega_B]\\
        &= \frac{1}{2} (\bb{E}[v(b_{\hat{\sigma}}(1))|\hat{\sigma}\in \Omega_A]+\bb{E}[v(b_{\hat{\sigma}}(1))|\hat{\sigma}\in \Omega_B]+\bb{E}[v(b_{\hat{\sigma}}(2))|\hat{\sigma}\in \Omega_B])\\
        &\qquad +\frac{1}{2} \sum_{i=2}^r \bb{E}[v(b_{\hat{\sigma}}(i))|\hat{\sigma}\in \Omega_A]+\frac{1}{2}\sum_{i=2}^{r-1} \bb{E}[v(b_{\hat{\sigma}}(i+1))|\hat{\sigma}\in \Omega_B].
    \end{align*}
    We thus have
    \begin{align*}
        \bb{E}[v_B(\hat{\sigma})|\hat{\sigma}\in \Omega]
        &\leq \left(1+\frac{1}{\sqrt{2}}\right)\big(\bb{E}[v(a_{\hat{\sigma}}(1))|\hat{\sigma}\in \Omega]\big) +\sum_{i=2}^r \bb{E}[v(a_{\hat{\sigma}}(i))|\hat{\sigma}\in \Omega]\\
        &\leq \left(1+\frac{1}{\sqrt{2}}\right)\big(\bb{E}[v_A(\hat{\sigma})|\hat{\sigma}\in \Omega]\big).
    \end{align*}
    Here, the first inequality follows from Observation~\ref{obs:RR-late} and Lemma~\ref{lm:RR-early}.
\end{proof}

Combining Theorem~\ref{thm:loweraddRRR} and Theorem~\ref{thm:upperaddRRR} gives our main result of this section (Theorem~\ref{thm:rrradd}).


\section{Beyond Additive Valuations: Submodular, XOS, and \\Subadditive}\label{sec:nonadditive}
In this section, we analyze the envy-ratio of randomized round-robin and iterated-RSD in settings where the agents' valuations belong to the much more general class of subadditive functions.

We begin with the case of submodular valuations, for which we show that both extensions of RSD continue to guarantee a constant envy-ratio. Recall that a valuation function $v$ is submodular if it has weakly decreasing marginal values. That is, for any pair $S,T$ of item sets with $S\subseteq T\subseteq M$ and any item $j\in M\setminus T$, we have $v(S\cup\{j\})-v(S) \geq v(T\cup\{j\})-v(T)$. We denote the marginal value of an item $j$ given a set $S$, i.e. the quantity $v(S\cup\{j\})-v(S)$, by $v_S(j)$.

The constructions that yield our lower bounds for submodular valuations differ substantially from previous examples and require a more extensive case analysis. It is therefore convenient to describe agents’ preferences using a geometric coverage representation. We associate each item $j\in M$ with an interval $I_j$. The valuation $v(S)$ is equal to 
the total length of the union of intervals $I_j$ for all $j\in S$. It is easy to see that valuations defined in this manner are monotone and submodular, since the marginal contribution of an item is the measure of newly covered length, and this can only decrease as the covered region grows. We first focus on the iterated-RSD mechanism with submodular valuations.

\thmirsdsub*

We divide the statement of the above theorem into separate lemmas for the lower bound and the upper bound.

\begin{lm}\label{lem:irsdsublower}
    When the agents have submodular valuations, the envy-ratio of iterated-RSD is at least $\frac{4}{7} (1+2\sqrt{2}) \approx 2.188$.
\end{lm}

\begin{proof}
    Using our geometric coverage representation, consider the following instance with $n+2$ agents and $(n+2)(2\ell+1)$ items for some large positive integer $\ell$. The instance has two highlighted agents $A$ and $B$, and a set $\mathcal{C}$ of $n$ other agents $C_1,\ldots,C_n$. It has two items labeled $x_1$ and $x_2$ that each have a singleton value of 1. It has two sets $Y=\{y_1,\ldots,y_\ell\}$ and $Y'=\{y'_1,\ldots,y'_\ell\}$ of $\ell$ items each. Every item in each of these sets has a singleton value of $1/2\ell$. Additionally, the pair $y_i$ and $y'_i$ are substitutes for each $i$. It has a set $Z={z_1,\ldots,z_\ell}$ of items, each of singleton value $1/2\ell$. Finally, it has a set $U=\set{u_1,\dots, u_{(n+2)(2\ell+1)-(3\ell+2)}}$ of 0-valued items for a total of $2\ell+1$ rounds.

    The intervals associated with both $A$'s and $B$'s valuations are as follows. $x_1$ is represented by the interval $[0,1]$ and $x_2$ is represented by $[1,2]$. $y_i$ and $y'_i$ are represented by $[(i-1)/2\ell,i/2\ell]$ for each $i\in[\ell]$. The item $z_i$ is represented by $[(1/2)+((i-1)/2\ell),(1/2)+(i/2\ell)]$ for each $i\in[\ell]$. Observe that if either agent selects $x_1$, then the items in the sets $Y$, $Y'$ and $Z$ have a marginal value of 0 in future rounds. The items in $U$ can be represented by $[0,0]$. Via an appropriate perturbation of the valuations, we may assume that if $A$ is breaking ties between multiple items with the same marginal value, then it selects items with the following preference: $U \succ x_1 \succ x_2 \succ Y \succ Y' \succ Z$. If $B$ is breaking ties between multiple items with the same marginal value, then it selects items with the following preference: $x_2 \succ x_1 \succ Z \succ Y' \succ Y \succ U$.

    Suppose the remaining agents $C_i$ are identical and additive, and for some $p\in(0,1)$ prefer the first $\floor{pn}$ items in $U$, i.e., $\{u_1,\ldots,u_{\floor{pn}}\}$, followed by the item $x_1$, followed by the remaining items in $U$. We have the following cases depending on the position of $A$ and $B$ in the \emph{first} round:
    \begin{itemize}
        \item $A$ is after less than $\floor{pn}$ of the agents in $\mathcal{C}$ in the first round. This happens with probability approximately $p$. In this case, $A$ selects $x_1$ and $B$ selects $x_2$ in the first round. In subsequent rounds, $B$ always picks all of $Z$ and $Y'$, while $A$ obtains items of zero value. Thus $A$ gets a value of $1$ and $B$ gets a value of $2$.
        \item $A$ is after $\floor{pn}$ of the agents in $\mathcal{C}$ and before $B$ in the first round. This happens with probability approximately $\frac{(1-p)^2}{2}$. In this case, $A$ gets $x_2$ in the first round and $B$ selects an item from $Z$. In subsequent rounds, $A$ selects all of $Y$ and $B$ selects the rest of $Z$ and all of $Y'$. Thus $A$ gets a value of $3/2$ and $B$ gets a value of $1$.
        \item $A$ is after $\floor{pn}$ of the agents in $\mathcal{C}$ and after $B$ in the first round. This event happens with probability approximately $p(1-p)+\frac{(1-p)^2}{2}$, where the first term corresponds to the event of $B$ appearing after less than $\floor{pn}$ of the agents in $\mathcal{C}$, and the second term corresponds to the remaining event. Here, $A$ selects all of $Y$ in the first $\ell$ rounds, and $B$ selects $x_2$ and $\ell-1$ items of $Z$. Additionally, $A$ is able to select at most one item from $Z$ if it appears before $B$ in the subsequent round. Then, $B$ selects all of the set $Y'$. Thus for large $\ell$, $A$ gets a value of approximately $1/2$, and $B$ gets a value of approximately $2$.
    \end{itemize}

    For large $\ell$, the envy-ratio is then lower-bounded by approximately
    \begin{align*}
        \frac{2(p + p(1-p)+\frac{(1-p)^2}{2}) + 1(\frac{(1-p)^2}{2})}
        {1(p) + \frac{3}{2}(\frac{(1-p)^2}{2}) + \frac{1}{2}(p(1-p)+\frac{(1-p)^2}{2})} = \frac{-p^{2}+2p+3}{p^{2}-p+2},
    \end{align*}
    which reaches a maximum value of $\frac{4}{7}(1+2\sqrt{2}) \approx 2.188$ at $p=4\sqrt{2}-5$.
\end{proof}

Next, we will prove a lemma for the upper bound. Before doing so, we first prove the following lemma, which is a consequence of submodularity. We require this lemma for our upper bounds for both the iterated-RSD mechanism and randomized round-robin mechanism.

\begin{lm}\label{lm:subm}
    Let $v$ be a submodular valuation, let $S$ be a set of items, and $(a_1, \dots, a_k)$, $(b_1, \dots, b_k)$ two sequences of $k$ items such that for all $i \in \{1, \dots k\}$, 
    \[v_{S \cup \{a_1, \dots, a_{i-1}\}}(b_i) \leq \alpha \cdot v_{S \cup \{a_1, \dots, a_{i-1}\}}(a_i).\]
    Then 
    \[v(\{b_1, \dots, b_k\})\leq v(S \cup \{a_1, \dots, a_k\})+\alpha v_{S}(\{a_1,\dots,a_k\}).\]
\end{lm}
\begin{proof}
    \begin{align*}
        v(\{b_1, \dots, b_k\})&\leq v(S\cup \{a_1, \dots, a_k\} \cup \{b_1, \dots, b_k\} \})\\
        &= v(S \cup \{a_1, \dots, a_k\})+\sum_{i=1}^{k-1} v(b_i|S \cup \{a_1, \dots, a_k\} \cup \{b_1, \dots, b_{i-1}\})\\
        &\leq v(S \cup \{a_1, \dots, a_k\})+\sum_{i=1}^{k-1}  v(b_i|S \cup \{a_1, \dots, a_{i-1}\})\\
        &\leq v(S \cup \{a_1, \dots, a_k\})+\sum_{i=1}^{k-1}  \alpha v(a_i| S \cup \set{a_1,\dots, a_{i-1}})\\
        &= v(S \cup \{a_1, \dots, a_k\})+\alpha v_{S}(\{a_1,\dots,a_k\}).
    \end{align*}
The first inequality follows by monotonicity, the second by submodularity, and the third by assumption.
\end{proof}

We also use the following simple claim, which is a 
restatement of Claim~\ref{cl:iRSD} from the additive case with more general notation.

\begin{cl}\label{cl:iRSD-subm}
    During the iterated-RSD mechanism, for all $k$ with $1\leq k\leq r$, let $a(k)$ (resp. $b(k)$) be the random item that $A$ selects during round $k$, and let $S(k) = \{a(1), \dots , a(k-1)\}$ be the random bundle of $A$ at the start of the $k^\text{th}$ round. Then
    \[ \bb{E}[v_{S(k)}(b(k))] \leq \sqrt{2} \cdot \bb{E}[v_{S(k)}(a(k))].\]
\end{cl}

The upper bound is stated in the following lemma.

\begin{lm}\label{lem:irsdsubupper}
    When the agents have submodular valuations, the envy-ratio of iterated-RSD is at most $1+\sqrt{2} \approx 2.414$.
\end{lm}

\begin{proof}

By Claim~\ref{cl:iRSD-subm}, $v_{A_j}(b_{\sigma(i)})\leq \sqrt{2} v_{A_j}(a_{\sigma(i)})$. In particular, we can then apply Lemma~\ref{lm:subm} with $\alpha=\sqrt{2}$ to get:
\[\bb{E}[v_B(\hat{\sigma})|\hat{\sigma}\in \Omega]=\bb{E}[v(\overline{B}_1))|\hat{\sigma}\in \Omega]\leq (1+\sqrt{2}) \bb{E}[v_A(\hat{\sigma})|\hat{\sigma}\in \Omega]\]
\end{proof}
Together, Lemmas~\ref{lem:irsdsublower} and~\ref{lem:irsdsubupper} prove Theorem~\ref{thm:irsdsub}..

Next, we analyze the envy-ratio of randomized round-robin for submodular valuations. We prove the following theorem.
\thmrrrsub*

As before, we start with the lower bound.

\begin{lm}\label{lem:rrrsublower}
    When the agents have submodular valuations, the envy-ratio of randomized round-robin is at least $\frac{9}{4} = 2.25$.
\end{lm}

\begin{proof}
We create an instance with $2n+2$ agents and $5\cdot(2n+2)$ items inducing five rounds. The instance consists of agents $A$ and $B$, and of $n$ identical agents $C_1,\ldots,C_n$ and $n$ identical agents $D_1,\ldots,D_n$. We partition the set of items into ten groups. Of these ten groups, agents~$A$ and~$B$ are primarily interested in obtaining items from the first five groups, since these are the only items for which they have a positive value. The items within any group are substitutes for agent~$A$, i.e., $A$ is interested in obtaining at most one item from each group. The first group consists of two items $1$ and $1'$. The second group contains the single item $2$. The third group also contains only item $3$. The fourth group consists of three items $4$, $4'$ and $4''$. The fifth group contains two items $5$ and $5'$. The five remaining groups are carefully designed in a manner that results in a specific picking order for the other $2n$ agents. The sixth, seventh, and ninth groups each contain $n-1$ items, the eighth group contains $n$ items, and the tenth group contains the remaining items.

We next describe the agents' valuations using our geometric coverage representation. Agent~$A$'s valuation concerns the following intervals: $I^A_1 = I^A_{1'}=[0,1]$, $I^A_2=[1,2]$, $I^A_3=[2,3]$, $I^A_4=I^A_{4'}=I^A_{4''}=[3,4]$, and $I^A_5=I^A_{5'}=[4,5]$. The remaining items have value 0 and can be represented with the 0-measure interval $[0,0]$. It is easy to verify that by an appropriate perturbation, we may assume agent~$A$ has a specific preference order over the items when breaking ties between items of the same marginal value. If there are no items with positive marginal value remaining then $A$ selects an item from the tenth group. Otherwise, $A$ has the following preference order: $1 \succ_A 4' \succ_A 4'' \succ_A 4 \succ_A 5' \succ_A 5 \succ_A 1' \succ_A 2 \succ_A 3$. Once none of the above items remain, $A$ selects from any group.

Agent~$B$ has a singleton value of 1 for items $1$, $2$, $3$, $4$, $5$, and $1'$ and a value of 0 for every other item. Its valuation is represented by $I^B_1=[0,1]$, $I^B_2=[1,2]$, $I^B_3=[2,3]$, $I^B_4=[3,4]$, $I^B_5=[-0.5,0.5]$, and $I^B_{1'}=[4,5]$. The intervals $I^B_1$ and $I^B_5$ partially overlap. By an appropriate perturbation, we may assume agent~$B$ has the following preference order over the items when breaking ties between items of the same marginal value: $1 \succ_B 2 \succ_B 5 \succ_B 1' \succ_B 3 \succ_B 4$. If none of these items remain, $B$ selects an item from the tenth group; otherwise it selects from any group. 

For every $i\in[n]$, agent~$C_i$ has a singleton value of 1 for the items in the sixth and seventh groups and the items $1'$, $3$, $4'$ and $4''$, and a value of 0 for the remaining items. $C_i$'s valuation for items $1'$, $3$, $4'$ and $4''$ is represented by $I^C_{1'}=[0.5,1.5]$, $I^C_3=[0,1]$, $I^C_{4'}=[1,2]$, and $I^C_{4''}=[2,3]$. Every item in the sixth group is represented by the interval $[2,3]$, and every item in the seventh group is represented by the interval $[3,4]$. By an appropriate perturbation, we may assume agent~$C_i$ has the following preference order over the items when breaking ties between items of the same marginal value. If an item in the sixth group is available, $C_i$ takes that item. Otherwise, $C_i$ selects item $4'$ if it is available, or item $1'$ if $4'$ is not available. If both are unavailable, $C_i$ selects an item from the seventh group. If all of these items have been claimed, $C_i$ selects item $3$. If item $3$ is also unavailable, $C_i$ selects item $4''$. Note that at any point in time, this preference order is only applied towards the items of highest marginal value for agent $C_i$. If these items are unavailable, $C_i$ selects an item from the tenth group; otherwise, if none of these items remain, it selects from any group.

Finally, for every $i\in[n]$, agent~$D_i$ has a singleton value of 1 for the items in the eighth and ninth groups and the item $5'$, and a value of 0 for the remaining items. $D_i$'s valuation for item $5'$ is represented by $I^D_{5'}=[0,1]$. Every item in the eighth group is represented by the interval $[1,2]$, and every item in the ninth group is represented by the interval $[2,3]$. By an appropriate perturbation, we may assume agent~$D_i$ has the following preference order over the items when breaking ties between items of the same marginal value. If an item in the eighth group is available, $D_i$ takes that item. Otherwise, if an item in the ninth group is available, $D_i$ takes that item. Finally, if these items are all unavailable, $D_i$ selects item $5'$, and if $5'$ has been taken, $D_i$ selects an item from the tenth group. If none of these items remain, $D_i$ selects any item.

Table~\ref{tab:225vals} summarizes the list of intervals and tie-breaking choices corresponding to the geometric representation of all valuations in this instance.
\begin{table}[ht]
    \begin{center}
    \begin{tabular}{c|ccccccccc|}
        & \multicolumn{9}{c}{{\bf Item}} \\
        \cline{2-10}
        {\bf Agent} & $1$ & $1'$ & $2$ & $3$ & $4$ & $4'$ & $4''$ & $5$ & $5'$ \\
        \cline{1-10}
        {\tt A} & \scell{$[0,1]$\\$2$} & \scell{$[0,1]$\\$8$} & \scell{$[1,2]$\\$9$} &\scell{$[2,3]$\\$10$} &\scell{$[3,4]$\\$5$} &\scell{$[3,4]$\\$3$} &\scell{$[3,4]$\\$4$} &\scell{$[4,5]$\\$7$} &\scell{$[4,5]$\\$6$} \\
        \cline{2-10}
        {\tt B} & \scell{$[0,1]$\\$1$} & \scell{$[4,5]$\\$4$} & \scell{$[1,2]$\\$2$} &\scell{$[2,3]$\\$5$} &\scell{$[3,4]$\\$6$} &\scell{$[0,0]$\\$-$} &\scell{$[0,0]$\\$-$} &\scell{$[-0.5,0.5]$\\$3$} &\scell{$[0,0]$\\$-$} \\
        \cline{2-10}
        {\tt C$_i$} & \scell{$[0,0]$\\$-$} & \scell{$[0.5,1.5]$\\$3$} & \scell{$[0,0]$\\$-$} &\scell{$[0,1]$\\$5$} &\scell{$[0,0]$\\$-$} &\scell{$[1,2]$\\$2$} &\scell{$[2,3]$\\$6$} &\scell{$[0,0]$\\$-$} &\scell{$[0,0]$\\$-$} \\
        \cline{2-10}
        {\tt D$_i$} & \scell{$[0,0]$\\$-$} & \scell{$[0,0]$\\$-$} & \scell{$[0,0]$\\$-$} &\scell{$[0,0]$\\$-$} &\scell{$[0,0]$\\$-$} &\scell{$[0,0]$\\$-$} &\scell{$[0,0]$\\$-$} &\scell{$[0,0]$\\$-$} &\scell{$[0,1]$\\$3$} \\
        \cline{1-10}
    \end{tabular}
    \vspace{5mm}
    
    \begin{tabular}{c|ccccc|}
        & \multicolumn{5}{c}{{\bf Item}} \\
        \cline{2-6}
        {\bf Agent} & \scell{$n-1$ items in\\sixth group} & \scell{$n-1$ items in\\seventh group} & \scell{$n$ items in\\eighth group} & \scell{$n-1$ items in\\ninth group} & \scell{$6n+4$ items in\\tenth group} \\
        \cline{1-6}
        {\tt A} & \scell{$[0,0]$\\$-$} & \scell{$[0,0]$\\$-$} & \scell{$[0,0]$\\$-$} & \scell{$[0,0]$\\$-$} & \scell{$[0,0]$\\$1$} \\
        \cline{2-6}
        {\tt B} & \scell{$[0,0]$\\$-$} & \scell{$[0,0]$\\$-$} & \scell{$[0,0]$\\$-$} & \scell{$[0,0]$\\$-$} & \scell{$[0,0]$\\$7$} \\
        \cline{2-6}
        {\tt C$_i$} & \scell{$[3,4]$\\$1$} & \scell{$[4,5]$\\$4$} & \scell{$[0,0]$\\$-$} & \scell{$[0,0]$\\$-$} & \scell{$[0,0]$\\$7$} \\
        \cline{2-6}
        {\tt D$_i$} & \scell{$[0,0]$\\$-$} & \scell{$[0,0]$\\$-$} & \scell{$[1,2]$\\$1$} & \scell{$[2,3]$\\$2$} & \scell{$[0,0]$\\$4$} \\
        \cline{1-6}
    \end{tabular}
    \end{center}
    
    \caption{Geometric representation of valuations in the example of Lemma~\ref{lem:rrrsublower}. Each cell has an interval and a picking index; if multiple items have the same marginal value, the agent selects from those items an item with the smallest picking index. A picking index of `$-$' indicates that the item has the lowest priority.}
    \label{tab:225vals}
\end{table}

We now analyze the execution of randomized round-robin on this instance. Let $C^*$ and $D^*$ respectively be the agents that appear last in the sets $\{C_1,\ldots,C_n\}$ and $\{D_1,\ldots,D_n\}$ in the permutation selected by randomized round-robin. We split our analysis into two cases, depending on whether $A$ or $B$ appears earlier in the ordering. In both cases, with high probability, agent~$C^*$ and agent~$D^*$ appear after both $A$ and $B$. Note that the relative order of appearance of agents $C^*$ and $D^*$ does not affect the allocation selected by the mechanism. Table~\ref{tab:225exec} shows the assignment made by randomized round-robin in each round, in both cases.
\begin{table}[]
    \centering
    \begin{minipage}{.4\linewidth}
      \centering
         \begin{tabularx}{0.75\textwidth}{c|XXXX}
           & \multicolumn{4}{c}{{\bf Agent}} \\
           \cline{2-5}
        {\bf Round} & {\tt A} & {\tt B} & {\tt C$^*$} & {\tt D$^*$} \\
        \cline{1-5}
        {\tt 1} & $1$ & $2$ & $4'$ & $8$ \\
        {\tt 2} & $4''$ & $5$ & $3$ & $5'$ \\
        {\tt 3} & $10$ & $1'$ & $10$ & $10$ \\
        {\tt 4} & $10$ & $4$ & $10$ & $10$ \\
        {\tt 5} & $10$ & $10$ & $10$ & $10$ \\
    \end{tabularx}
    \end{minipage}
    \quad
    \begin{minipage}{.4\linewidth}
      \centering
         \begin{tabularx}{0.75\textwidth}{c|XXXX}
           & \multicolumn{4}{c}{{\bf Agent}} \\
           \cline{2-5}
        {\bf Round} & {\tt B} & {\tt A} & {\tt C$^*$} & {\tt D$^*$} \\
        \cline{1-5}
        {\tt 1} & $1$ & $4'$ & $1'$ & $8$ \\
        {\tt 2} & $2$ & $5'$ & $4''$ & $10$ \\
        {\tt 3} & $3$ & $10$ & $10$ & $10$ \\
        {\tt 4} & $4$ & $10$ & $10$ & $10$ \\
        {\tt 5} & $5$ & $10$ & $10$ & $10$ \\
    \end{tabularx}
    \end{minipage}
    \caption{Execution of randomized round-robin on the instance from Table~\ref{tab:225vals}, with two cases corresonding to the relative order of $A$ and $B$. Items from the sixth, seventh, eighth, ninth and tenth groups are labeled 6, 7, 8, 9, and 10 respectively.}
    \label{tab:225exec}
\end{table}

If $A$ appears before $B$, it selects item $1$ in the first round, and $B$ selects item $2$. Since the first $n-1$ agents in $\{C_1,\ldots,C_n\}$ select all $n-1$ items in the sixth group, when $C^*$ appears, it selects item $4'$. $D^*$ selects an item in the eighth group in the first round. Next, in round 2, $A$ selects item $4''$ and $B$ selects item $5$, with both cases having a marginal value of 1. Since the items in the seventh group are taken, and item $1'$ has a smaller marginal value of 0.5, $C^*$ selects item $3$. In this round, $D^*$ selects item $5'$ since the $n-1$ items in the ninth group are taken. Finally, in the remaining rounds, since no items of positive marginal value remain, $A$, $C^*$ and $D^*$ select items from the tenth group. However, in rounds 3 and 4, $B$ selects the items $1'$ and $4$ respectively.

For the other case, when $B$ appears before $A$, $B$ selects item $1$ in the first round and $A$ selects item $4'$. Since the sixth group has been fully assigned, and item $4'$ is unavailable, $C^*$ selects item $1'$. Agent~$D^*$ selects an item from the eighth group. Next, in the second round, $B$ selects item $2$ and $A$ selects item $1$, with both agents obtaining an item of marginal value 1. Since the seventh group has been fully assigned, and item~$3$ has a marginal value of $0.5$, $C^*$ selects item $4''$. Starting with round 2, $D^*$ selects only items in the tenth group. Finally, in the third round, $B$ selects item $3$. Starting with round 3 and until the end, both $A$ and $C^*$ select only items in the tenth group, since no other items of positive marginal value remain. In rounds~4 and~5, $B$ selects items $4$ and $5$.

Let $X_A$ and $X_B$ be the random bundles obtained by agents $A$ and $B$ respectively. When $n$ is large, the envy-ratio for this instance is bounded below by
\[
\frac{\mathbb{E}[v(X_B)]}{\mathbb{E}[v(X_A)]} \approx \frac{\frac{1}{2}\cdot4 + \frac{1}{2}\cdot5}{\frac{1}{2}\cdot2 + \frac{1}{2}\cdot2} = \frac{9}{4}.
\]
The completes the proof of the factor-$\frac94$ lower bound.
\end{proof}

Next, we prove the upper bound. Consider an execution of round-robin with a permutation $\sigma\in \Omega$. For all $k \in \{1, \dots, r\}$, let $a_\sigma(k)$ (resp. $b_\sigma(k)$) be the item that $A$ (resp. $B$) selects at round $k$, and let $S_\sigma(k)= \{a_\sigma(1), \dots, a_\sigma(k-1)\}$ be the set of item in $A$'s bundle before round $k$.

We analyze randomized round-robin for the case of $r \geq 2$ rounds. Consider an execution of round-robin with a permutation $\sigma\in \Omega$. For all $k \in \{1, \dots, r\}$, let $a_\sigma(k)$ (resp. $b_\sigma(k)$) be the item that $A$ (resp. $B$) selects at round $k$, and let $S_\sigma(k)= \{a_\sigma(1), \dots, a_\sigma(k-1)\}$ be the set of item in $A$'s bundle before round $k$. The following is a more general version of Observation~\ref{obs:RR-late} from the additive case.

\begin{obs}\label{obs:RR-late-subm}
     For all $k \in \{1, \dots, r\}$, if $\sigma \in \Omega_A$ then $v_{S_\sigma(k)}(b_\sigma(i))\leq v_{S_\sigma(k)}(a_\sigma(i))$. For all $k \in \{1, \dots, r-1\}$, if $\sigma \in \Omega_B$ then $v_{S_\sigma(k)}(b_\sigma(i+1))\leq v_{S_\sigma(k)}(a_\sigma(i))$.
\end{obs}

We also reuse the following lemma from the additive case.

\RRearly*

We are ready to prove our upper bound for randomized round-robin with submodular valuations.

\begin{lm}\label{lem:rrrsubupper}
    When the agents have submodular valuations, the envy-ratio of randomized round-robin is at most $2+(1/\sqrt{2}) \approx 2.707$.
\end{lm}

\begin{proof}

For all $\sigma \in \Omega$ and all $k \in {1,\dots,r}$, let $a_\sigma(k)$ (resp. $b_\sigma(k)$) denote the item selected by $A$ (resp. $B$) at round $k$, and let $S_\sigma(k)=\{a_\sigma(1),\dots,a_\sigma(k-1)\}$ be the set of items in $A$’s bundle before round $k$ when a round robin mechanism is run according to $\sigma$. Moreover let $S^A_\sigma$ (resp $S^B_\sigma$) be the bundle of $A$ (resp of $B$) when the round robin terminates. Let $\hat \sigma$ be the random permutation used by the randomized round robin mechanism. 

Using the first point of Observation~\ref{obs:RR-late-subm}, we get the following inequalities: for all $\sigma\in \Omega_A$ and for all $i\in \set{2,\dots,r}$,  $v_{S_\sigma(i)}(b_\sigma(i))\leq v_{S_\sigma(i)}(a_\sigma(i))$. Applying Lemma~\ref{lm:subm} with $\alpha = 1$, the set $\{a_\sigma(1)\}$, and sequences $\{b_\sigma(2), \dots, b_\sigma(r)\}$, $\{a_\sigma(2),\dots,a_\sigma(r)\}$, we obtain for all $\sigma\in \Omega_A$,
\[v(S_\sigma^B \setminus \{b_\sigma(1)\}) \leq v(S_\sigma^A) +  v_{S_\sigma(2)}(S_\sigma^A \setminus S_\sigma(2)).\]

Now taking expectation over $\Omega_A$, we obtain
\[\bb{E}[v(S_{\hat\sigma}^B \setminus \{b_{\hat\sigma}(1)\})| \hat \sigma \in \Omega_A] \leq \bb{E}[v(S_{\hat\sigma}^A)| \hat \sigma \in \Omega_A] + \bb{E}[v_{S_{\hat\sigma}(2)}(S_{\hat\sigma}^A \setminus S_{\hat\sigma}(2))| \hat \sigma \in \Omega_A].
\]

Using the second point of Observation~\ref{obs:RR-late-subm}, we get the following inequalities: for all $\sigma\in \Omega_B$ and for all $ i\in \set{2,\dots,r-1}$,  $v_{S_\sigma(i)}(b_\sigma(i+1))\leq v_{S_\sigma(i)}(a_\sigma(i))$. Applying the previous Lemma~\ref{lm:subm} with the set $\{a_\sigma(1)\}$, and sequences $\{b_\sigma(3), \dots, b_\sigma(r)\}$, $\{a_\sigma(2),\dots,a_\sigma(r-1)\}$, we obtain for all $\sigma\in \Omega_A$,
\[v(S_\sigma^B \setminus \{b_\sigma(1),b_\sigma(2)\}) \leq v(S_\sigma^A) + v_{S_\sigma(2)}(S_\sigma^A \setminus S_\sigma(2)).\]
Taking expectation over $\Omega_B$, we obtain
\[\bb{E}[v(S_{\hat\sigma}^B \setminus \{b_{\hat\sigma}(1), b_{\hat \sigma}(2)\})| \hat \sigma \in \Omega_B] \leq \bb{E}[v(S_{\hat\sigma}^A)| \hat \sigma \in \Omega_B] + \bb{E}[v_{S_{\hat\sigma}(2)}(S_{\hat\sigma}^A \setminus S_{\hat\sigma}(2))| \hat \sigma \in \Omega_B].
\]
Additionally, applying Lemma~\ref{lm:RR-early} we get:

$\begin{aligned}
\frac12\bb{E}[v(\{b_{\hat{\sigma}}(1)\}) \mid \hat{\sigma}\in\Omega_A]
+ \frac12\bb{E}[v(\{b_{\hat{\sigma}}(1)\}) \mid \hat{\sigma}\in\Omega_B] 
+ \frac12\bb{E}[v(\{b_{\hat{\sigma}}(2)\}) \mid \hat{\sigma}\in\Omega_B] \\
\le \left(1+\frac1{\sqrt2}\right)
\bb{E}[v(\{a_{\hat{\sigma}}(1)\}) \mid \hat{\sigma}\in\Omega].
\end{aligned}$

Using the linearity of expectation and submodularity of $v$, we have
\begin{align*}
    \bb{E}[v(S^B_{\hat \sigma})|\hat{\sigma}\in \Omega] & = \frac12 \bb{E}[v(\{b_{\hat\sigma}(1)\})| \hat \sigma \in \Omega_A] +\frac12\bb{E}[v_{\{b_{\hat\sigma}(1)\}}(S^B_{\hat\sigma} \setminus \{b_{\hat\sigma}(1)\} )| \hat \sigma \in \Omega_A]\\
    &\qquad +\frac12 \bb{E}[v(\{b_{\hat\sigma}(1)\})| \hat \sigma \in \Omega_B] + \frac12 \bb{E}[v_{\{b_{\hat\sigma}(1)\}}(\{b_{\hat\sigma}(2)\})| \hat \sigma \in \Omega_B] \\
     & \qquad + \frac12\bb{E}[v_{\{b_{\hat\sigma}(1), b_{\hat\sigma}(2)\}}(S^B_{\hat\sigma} \setminus \{b_{\hat\sigma}(1), b_{\hat\sigma}(2)\} )| \hat \sigma \in \Omega_B]\\
     &\leq \frac12 \bb{E}[v(\{b_{\hat\sigma}(1)\})| \hat \sigma \in \Omega_A] +\frac12\bb{E}[v(S^B_{\hat\sigma} \setminus \{b_{\hat\sigma}(1)\} )| \hat \sigma \in \Omega_A]\\
    &\qquad +\frac12 \bb{E}[v(\{b_{\hat\sigma}(1)\})| \hat \sigma \in \Omega_B] + \frac12 \bb{E}[v(\{b_{\hat\sigma}(2)\})| \hat \sigma \in \Omega_B] \\
     & \qquad + \frac12\bb{E}[v(S^B_{\hat\sigma} \setminus \{b_{\hat\sigma}(1), b_{\hat\sigma}(2)\} )| \hat \sigma \in \Omega_B]
\end{align*}
    Combining everything and rearranging the terms, we obtain:
\begin{align*}
    \bb{E}[v(S^B_{\hat \sigma})|\hat{\sigma}\in \Omega] 
    &\leq \left(\frac12 \bb{E}[v(\{b_{\hat\sigma}(1)\})| \hat \sigma \in \Omega_A] +\frac12 \bb{E}[v(\{b_{\hat\sigma}(1)\})| \hat \sigma \in \Omega_B] + \frac12 \bb{E}[v(\{b_{\hat\sigma}(2)\})| \hat \sigma \in \Omega_B]\right)\\
    &\qquad + \frac12\bb{E}[v(S^B_{\hat\sigma} \setminus \{b_{\hat\sigma}(1)\} )| \hat \sigma \in \Omega_A]   + \frac12\bb{E}[v(S^B_{\hat\sigma} \setminus \{b_{\hat\sigma}(1), b_{\hat\sigma}(2)\} )| \hat \sigma \in \Omega_B]\\
    &\leq  \left(1+\frac1{\sqrt2}\right)
    \bb{E}[v(\{a_{\hat{\sigma}}(1)\}) \mid \hat{\sigma}\in\Omega]\\
    &\qquad + \frac{1}{2}\bb{E}[v(S_{\hat\sigma}^A)| \hat \sigma \in \Omega_A] + \frac{1}{2}\bb{E}[v_{S_{\hat\sigma}(2)}(S_{\hat\sigma}^A \setminus S_{\hat\sigma}(2))| \hat \sigma \in \Omega_A]\\
    &\qquad +\frac{1}{2}\bb{E}[v(S_{\hat\sigma}^A)| \hat \sigma \in \Omega_B] + \frac{1}{2}\bb{E}[v_{S_{\hat\sigma}(2)}(S_{\hat\sigma}^A \setminus S_{\hat\sigma}(2))| \hat \sigma \in \Omega_B]\\
    &\leq  \bb{E}[v(S_{\hat\sigma}^A)| \hat \sigma \in \Omega]\\
    &\qquad +\frac{1}{2}\left(1+\frac1{\sqrt2}\right)
    \bb{E}[v(\{a_{\hat{\sigma}}(1)\}) \mid \hat{\sigma}\in\Omega_A]+\frac{1}{2}\bb{E}[v_{S_{\hat\sigma}(2)}(S_{\hat\sigma}^A \setminus S_{\hat\sigma}(2))| \hat \sigma \in \Omega_A]\\
    &\qquad +\frac{1}{2}\left(1+\frac1{\sqrt2}\right)
    \bb{E}[v(\{a_{\hat{\sigma}}(1)\}) \mid \hat{\sigma}\in\Omega_B]+ \frac{1}{2}\bb{E}[v_{S_{\hat\sigma}(2)}(S_{\hat\sigma}^A \setminus S_{\hat\sigma}(2))| \hat \sigma \in \Omega_B]\\
    &\leq  \bb{E}[v(S_{\hat\sigma}^A)| \hat \sigma \in \Omega] +\frac{1}{2}\left(1+\frac1{\sqrt2}\right)
    \bb{E}[v(S_{\hat\sigma}^A)| \hat \sigma \in \Omega_A] +\frac{1}{2}\left(1+\frac1{\sqrt2}\right)
    \bb{E}[v(S_{\hat\sigma}^A)| \hat \sigma \in \Omega_B]\\
    &= \left(2+\frac{1}{\sqrt{2}}\right)\bb{E}[v(S_{\hat\sigma}^A)| \hat \sigma \in \Omega]
\end{align*}
\end{proof}

Together, Lemmas~\ref{lem:rrrsublower} and~\ref{lem:rrrsubupper} prove Theorem~\ref{thm:rrrsub}.

Finally, we analyze the two more general valuation classes. We show a sharp boundary between submodular and XOS valuations: as soon as the agents' valuations are allowed to belong to the wider class, the envy-ratio becomes unbounded in $|N|$. Nevertheless, we present a tight analysis of this ratio for both XOS and subadditive valuations, showing that the envy-ratio of both mechanisms is exactly $r$, the number of rounds in the execution of each mechanism. 

\thmgeneral*

We split the proof into two lemmas: one establishes the upper bound for subadditive valuations, and the other establishes the lower bound for XOS valuations.

\begin{lm}
    For subadditive valuations, the envy-ratio of both randomized round-robin and iterated RSD is at most $r$.
\end{lm}

\begin{proof}
Fix an instance with agents $A$ and $B$, and a set $\mathcal{C}$ consisting of the $n$ remaining agents. Consider either the randomized round-robin mechanism or the iterated-RSD mechanism. We assume that $A$ has a positive value for some subset of items, which by subadditivity implies that $A$ has a positive value for at least one singleton.

Let $X_A$ and $X_B$ be the random bundles allocated to $A$ and $B$ respectively by the mechanism.
We reuse the notation $\Omega$, $\Omega'$, and $\Omega(\sigma')$ for $\sigma' \in \Omega'$. Let $\hat \sigma$ be the random permutation used in the first round of the mechanism.

Because $A$ has a positive value for at least one singleton, for each $\sigma' \in \Omega'$, there exists $\sigma \in \Omega(\sigma')$ such that $\mathbb{E}[v(X_A) | \hat \sigma = \sigma]>0$, and therefore $\mathbb{E}[v(X_A) | \hat \sigma \in \Omega(\sigma')]>0$. We can imitate the proof of Lemma~\ref{lm:order} and we get:

\[\frac{\mathbb{E}[v(X_B)]}{\mathbb{E}[v(X_A)]} \leq \max_{\sigma' \in \Omega'} \frac{\mathbb{E}[v(X_B)|\hat \sigma \in \Omega(\sigma')]}{\mathbb{E}[v(X_A)|\hat \sigma \in \Omega(\sigma')]}.\]

\smallskip We fix a permutation $\sigma' \in \Omega'$ of the agents in $\mathcal{C}$. Without loss of generality we label the agents of $\mathcal{C}$ such that $\sigma' = (C_1, \dots,C_n)$. Let $(i,j)$ be an ordered pair with $0\leq i \leq j$, and consider the permutations $\sigma_A(i,j)$, and $\sigma_B(i,j)$ as in the proof of Lemma~\ref{lm:c'upper_bound}. 

We analyze an execution in which the first round of serial dictatorship is run according to a permutation $\sigma \in \Omega$. Let $a_\sigma$ and $b_\sigma$ denote the values, according to $A$’s valuation, of the singleton sets containing the items allocated to $A$ and $B$, respectively, in this round. Let $c_\sigma$ denote the maximum value, according to $A$’s valuation, of any singleton set corresponding to an item that remains unallocated at the end of the first round. 

$A$ gets a singleton of value $a_\sigma$ in the first round, therefore the value of $A$ for its bundle is at least $a_\sigma$ by the monotonicity of its valuation function. $B$ gets in the first round a singleton of value $b_\sigma$ for $A$, and at each later round, the valuation of $B$'s bundle according to $A$ increases by at most $c_{\sigma}$ by the subadditivity of $A$'s valuation. Therefore $A$'s value for $B$'s bundle when the mechanism terminates is at most $b_\sigma + (r-1) \cdot c_\sigma$.

Note that the partition $\Omega(\sigma') = \bigsqcup_{0 \leq i \leq j \leq n} \{\sigma_A{(i,j), \sigma_B(i,j)\}}$ implies that on one side, 

\[\mathbb{E}[v(X_A)|\hat \sigma \in \Omega(\sigma')] \geq  \frac{1}{2\cdot {{n+2} \choose 2}}\sum_{0 \leq i \leq j \leq n} (a_{\sigma_A(i,j)} + a_{\sigma_B(i,j)}).\]

On the other side, \[\mathbb{E}[v(X_B)|\hat \sigma \in \Omega(\sigma')] \geq  \frac{1}{2\cdot {{n+2} \choose 2}}\sum_{0 \leq i \leq j \leq n} (b_{\sigma_A(i,j)} + b_{\sigma_B(i,j)} + (r-1)\cdot (c_{\sigma_A(i,j)}+c_{\sigma_B(i,j)})).\]

The following two claims provide the inequalities we need to conclude.

\begin{cl}\label{cl:subadd1}
    For all $(i,j)$ with $0 \leq i \leq j \leq n$, $a_{\sigma_A(i,j)} \geq \max\{b_{\sigma_A(i,j)},b_{\sigma_B(i,j)},c_{\sigma_A(i,j)},c_{\sigma_B(i,j)}\}$.
\end{cl}
\begin{proof}
$\sigma_A(i,j)$ correspond to the permutation where $A$ selects its favorite singleton after the $i$ first agents of $\mathcal{C}$. At that point, the item chosen by $B$ in $\sigma_B(i,j)$ as well as the item chosen by $B$ in $\sigma_A(i,j)$ are still available. Similarly, the items that remain unallocated at the end of the first round in both $\sigma_B(i,j)$ and $\sigma_A(i,j)$ are still available, which prove the claim.
\end{proof}

\begin{cl}\label{cl:subadd2}
    For all $(i,j)$ with $0 \leq i \leq j \leq n$, $a_{\sigma_B(i,j)} \geq \min\{b_{\sigma_A(i,j)},c_{\sigma_A(i,j)}\}$.
\end{cl}
\begin{proof}   
The argument closely follows the proof of Lemma~\ref{lm:z}. Suppose the agents in $\mathcal{C}$ select items via a serial dictatorship according to $\sigma'$ (in the absence of $A$ and $B$), and let $z_\ell$ denote the item chosen by agent $C_\ell$. For each $\ell\in\{0,\ldots,n\}$, let $x_\ell$ and $y_\ell$ be, respectively, the most-preferred and second most-preferred singletons of agent $A$ in the set $M \setminus \{z_1, \dots, z_\ell\}$.

In $\sigma_B(i,j)$, agent $A$ picks after $j+1$ items have been removed. By Lemma~\ref{lm:airplane}, all items in ${z_1,\dots,z_j}$ have been removed, along with one additional item. Therefore, $A$ receives either $x_j$ or $y_j$, and hence $a_{\sigma_B(i,j)} \ge v_A(\{y_j\})$.

Similarly, in $\sigma_A(i,j)$, agent $B$ picks after $j+1$ items have been removed. By Lemma~\ref{lm:airplane}, all items in $\{z_1,\dots,z_j\}$ have been removed, and the only singleton better than $y_j$ that may still be available is $x_j$. If $B$ does not pick $x_j$, then $a_{\sigma_B(i,j)} \ge b_{\sigma_A(i,j)}$. Otherwise, if $B$ picks $x_j$, then at the end of the round the best remaining singleton is worth at most $y_j$, and thus $a_{\sigma_B(i,j)} \ge c_{\sigma_A(i,j)}$.
\end{proof}

We now combine the two claims. For each $(i,j)$, Claim~\ref{cl:subadd1} implies
\[
b_{\sigma_A(i,j)} \le a_{\sigma_A(i,j)},\qquad c_{\sigma_A(i,j)} \le a_{\sigma_A(i,j)},\qquad 
b_{\sigma_B(i,j)} \le a_{\sigma_A(i,j)},\qquad c_{\sigma_B(i,j)} \le a_{\sigma_A(i,j)}.
\]
Moreover, Claim~\ref{cl:subadd2} implies $b_{\sigma_A(i,j)} \le a_{\sigma_B(i,j)}$ or $c_{\sigma_A(i,j)} \le a_{\sigma_B(i,j)}$, and in either case
\[
b_{\sigma_A(i,j)} + c_{\sigma_A(i,j)} \le a_{\sigma_A(i,j)} + a_{\sigma_B(i,j)}.
\]
Additionally, it is easy to see that $c_{\sigma_B(i,j)} \le a_{\sigma_B(i,j)}$, because the items that remain unallocated at the end of the first round in $\sigma_B(i,j)$ are available when $A$ selects its favorite item. 
Hence, for every $(i,j)$, we have
\begin{alignat*}{4}
    &b_{\sigma_A(i,j)} + b_{\sigma_B(i,j)} &&+ (r-1)\bigl(c_{\sigma_A(i,j)} + c_{\sigma_B(i,j)}\bigr)&&&\\
    =\quad&\bigl(b_{\sigma_A(i,j)} + c_{\sigma_A(i,j)}\bigr) &&+ b_{\sigma_B(i,j)} + (r-2)c_{\sigma_A(i,j)} &&+ (r-1)c_{\sigma_B(i,j)}&\\
    \le\quad &\bigl(a_{\sigma_A(i,j)} + a_{\sigma_B(i,j)}\bigr) &&+ (r-1)a_{\sigma_A(i,j)} &&+ (r-1)a_{\sigma_B(i,j)}&\\
    =\quad&r\bigl(a_{\sigma_A(i,j)} + a_{\sigma_B(i,j)}\bigr).
\end{alignat*}
Summing over all pairs $(i,j)$ and using the two displayed bounds on the conditional expectations yields
\[
\frac{\mathbb{E}[v(X_B)\mid \hat\sigma \in \Omega(\sigma')]}{\mathbb{E}[v(X_A)\mid \hat\sigma \in \Omega(\sigma')]}
\le r.
\]
Since this holds for every $\sigma' \in \Omega'$, we obtain
\[
\frac{\mathbb{E}[v(X_B)]}{\mathbb{E}[v(X_A)]}
\le \max_{\sigma' \in \Omega'} 
\frac{\mathbb{E}[v(X_B)\mid \hat\sigma \in \Omega(\sigma')]}{\mathbb{E}[v(X_A)\mid \hat\sigma \in \Omega(\sigma')]}
\le r.
\]
\end{proof}

\begin{lm}
    For XOS valuations, the envy-ratio of both randomized round-robin and iterated-RSD is at least $r$.
\end{lm}

\begin{proof}
    Recall that a valuation function $v$ is XOS if there exists a nonempty collection $v_1,\ldots, v_t$ of additive set functions such that for any $S\subseteq M$, $v(S) = \max_{i\in[t]}v_i(S)$.

    The instance that gives the lower bound is simple. It consists of two agents $A$ and $B$, and two sets of items $P = \{p_1,\ldots,p_r\}$ and $Q=\{q_1,\ldots,q_r\}$. It is easy to see that the mechanisms terminate in $r$ rounds. Agent~$A$'s XOS valuation function is the maximum of the additive functions $a_1$ and $a_2$. The function $a_1$ assigns a value of $1$ to $p_1$, a value of $\varepsilon$ to every other item in $P$, and a value of $0$ to every item in $Q$. The function $a_2$ assigns a value of $0$ to every item in $P$ and a value of $1-\varepsilon$ to every item in $Q$. Agent~$B$ simply has an additive valuation that assigns a value of $0$ to every item in $P$ and a value of $1$ to every item in $Q$.

    For both mechanisms, in any permutation, $A$ selects $p_1$ in the first round. In subsequent rounds, the marginal value of every item in $Q$ remains equal to 0, and $A$ selects the remaining items in $P$. On the other hand, $B$ only selects items from $Q$. When either mechanism terminates, $A$ is allocated the set $P$ (of value $1+(r-1)\varepsilon$) and $B$ is allocated the set $Q$ (of value $r(1-\varepsilon)$ for $A$). Consequently, with $\varepsilon\rightarrow0$ the envy-ratio approaches $r$.
\end{proof}

\section{Conclusion}\label{sec:discussion}
In this work, we settle a fundamental open question about the fairness guarantees of random serial dictatorship. We show that this mechanism is exactly $\sqrt{2}$-approximately envy-free in expectation for the house allocation problem. We also analyze the two natural extensions of this mechanism to more general settings, for which we establish tight or nearly-tight bounds on the envy-ratio for a variety of important valuation classes. These results constitute the first quantitative guarantees on the fairness of random serial dictatorship and its natural generalizations. Some of our theorems leave open a small constant-factor gap between the upper and lower bounds on the envy-ratios; a natural open problem is to determine the exact envy-ratio of each mechanism for the corresponding valuation classes for these theorems. Additionally, our goal in this work was to analyze the intrinsic fairness properties of these mechanisms by considering non-strategic agents that always select an item with the highest marginal value. While this model is well-justified for RSD, which is strategyproof for the house allocation problem, it abstracts away from strategic considerations that arise for randomized round-robin and iterated-RSD, which are both known to be non-strategyproof in multi-round settings. Understanding the extent to which similar guarantees can be obtained under strategic behavior remains an intriguing open problem.

\bibliographystyle{ACM-Reference-Format}
\bibliography{ref}

\end{document}